\documentclass{lmcs}
\pdfoutput=1

\usepackage{lastpage}
\lmcsdoi{19}{4}{33}
\lmcsheading{}{\pageref{LastPage}}{}{}%
{Aug.~20,~2018}{Dec.~20,~2023}{}


\usepackage{macros_perso}
\usepackage{hyperref}

\theoremstyle{plain} 

\begin{document}

\title[Synchronizability of CFSM is not Decidable]{Synchronizability of Communicating\texorpdfstring{\\}{ }Finite State Machines is not Decidable}

\titlecomment{{\lsuper*}Extended version of an article published in ICALP'17 proceedings}

\author[A.~Finkel]{Alain Finkel\lmcsorcid{0000-0003-0702-3232}}[a]
\author[E.~Lozes]{Etienne Lozes\lmcsorcid{0000-0001-8505-585X}}[b]	

\address{Universit\'e Paris-Saclay, CNRS, ENS Paris-Saclay, Institut Universitaire de France, Laboratoire Méthodes Formelles, 91190, Gif-sur-Yvette, France}
\urladdr{https://ens-paris-saclay.fr/alain-finkel/}
\email{alain.finkel@ens-paris-saclay.fr}  

\address{Université Côte d'Azur, CNRS, I3S, France}
\email{elozes@i3s.unice.fr}  




\begin{abstract}
  A system of communicating finite state machines is \emph{synchronizable}
  if its send trace semantics, i.e.
  the set of sequences of sendings it can perform, is the same when its
  communications are
  FIFO asynchronous and when they are just rendez-vous synchronizations.
  This property was claimed to be decidable in several conference and journal
  papers
  for either mailboxes or peer-to-peer communications,
  thanks to a form of small model property. In this paper, we show that
  this small model property does not hold neither for mailbox communications,
  nor for peer-to-peer communications, therefore the decidability
  of synchronizability becomes an open question. We close this question for
  peer-to-peer communications, and we show that
  synchronizability is actually undecidable. We show that synchronizability is
  decidable if the topology of communications is an oriented ring. We also
  show that,
  in this case,
  synchronizability implies the absence of unspecified receptions, and the
  channel-recognizability of the reachability set.
\end{abstract}

\maketitle


\section{Introduction\label{sec:intro}}
Asynchronous distributed systems are error prone not only because they are
difficult to program, but also because they are difficult
to execute in a reproducible way.
The slack of communications,
measured by the number of messages that can be buffered in a same
communication channel, is not always under
the control of the programmer, and even when it is, it may be delicate to
choose the right size of the communication buffers.

The synchronizability of a system of communicating machines
is a property introduced by Basu and Bultan~\cite{DBLP:journals/tse/FuBS05,BasuB11} that formalizes
the idea of a distributed system that is ``slack elastic'',
in the sense that its behaviour is the same whatever the size of the buffers,
and in particular it is enough to detect bugs by considering executions
with buffers of size one~\cite{BasuBO12,BasuB16}. Synchronizability can
also be used for checking other properties like
choreography realizability~\cite{poplBasuBO12}.

For instance, consider the machines
$$
P\ =\ !a\cdot{}!b \qquad \mbox{and} \qquad
Q\ =\ ?a\cdot{}?b
$$
where $P$ may send two messages $a$ and $b$ in sequence to $Q$, and $Q$ is
ready to receive them.
These two machines
form a synchronous system $P|Q$:
the asynchronous trace $!a\cdot{}!b\cdot{}?a\cdot{}?b$ is
``equivalent'' to the synchronous trace $!a\cdot{}?a\cdot{}!b\cdot{}?b$.
Two traces
are considered ``equivalent''  by Basu and Bultan if they present the same
sequence of send actions, i.e. that they are identical after erasing all
receive actions. This is the case for the above example, as both traces
result in $!a\cdot{}!b$ after erasing the receive actions. A system is language
synchronizable if all of its traces are equivalent to a synchronous trace.
An additional requirement is that
two ``equivalent'' traces lead to the same
configuration; when it is the case, the system is called synchronizable.
For instance, taking
$$
P\ =\ !a\cdot{}?b \qquad \mbox{and} \qquad
Q\ =\ ?a\cdot{}!b\ +\ !b\cdot{} ?a
$$
$P|Q$ is language synchronizable but it is not synchronizable, because the asynchronous trace $!a\cdot{}!b\cdot{}?a\cdot{}?b$ does not lead to the same
configuration as the synchronous trace $!a\cdot{}?a\cdot{}!b\cdot{}?b$.

For systems with more than two machines, there are at least two
distinct reasonable semantics of a system
of communicating machines with FIFO queues: either each message sent
from $P$ to $Q$ is stored in a queue which is specific to the pair $(P,Q)$,
which we will call the ``peer-to-peer'' semantics, or all messages sent to
$Q$ from several other peers are mixed toghether in a queue that is specific
to $Q$, which we will call the ``mailbox'' semantics.

Basu and Bultan claimed that synchronizability is decidable,
first for the mailbox semantics~\cite{BasuBO12}, and later for other
semantics, including the peer-to-peer one~\cite{BasuB16}. Their
main argument was a small model property, stating that if all 1-bounded
traces are equivalent to synchronous traces then the system is synchronizable.

This paper corrects some of these claims and discuss some related questions.


\begin{itemize}
\item We establish the undecidability of synchronizability for systems of $3$-CFSMs with the peer-to-peer semantics
(Theorem~\ref{thm:undecidability}). This shows that the claim on the decidability of synchronizability for the peer-to-peer semantics ~\cite{BasuB16} is actually wrong.
\item We provide counter-examples to the small model property for systems of $3$-CFSMs
both for the peer-to-peer semantics (Example~\ref{ex:counter-example-peer-to-peer})
and the mailbox semantics (Example~\ref{ex:counter-example-mailbox}) which illustrate that the claims in~\cite{BasuBO12,BasuB16} are not proved correctly. The fact that the small model property is false for the peer-to-peer semantics is a consequence of the previous result but this is not a consequence for the mailbox semantics.
\item We prove that the small model property is true for systems of $2$-CFSMs and more generally for systems of communicating machines on an oriented ring both under the
mailbox and peer-to-peer semantics (actually both are the same in that case) and therefore synchronizability is decidable for oriented rings (Theorem~\ref{thm:decidability}).
\item We show that the reachability set of synchronizable systems
is channel-recognizable (Theorem~~\ref{thm:channel-rec}), ie the set of
reachable configurations is regular.
\item Finally, we show that the counter-examples we gave invalidate other
claims, in particular a result used for
checking \emph{stability}~\cite{AkrounSY16,AkrounS18}.
\end{itemize}


\subsection*{Outline} The paper first
focuses on the peer-to-peer communication model.
Section~\ref{sec:prel} introduces all notions of
communicating finite state machines and synchronizability. In
Section~\ref{sec:undecidability}, we show that synchronizability is
undecidable. Section~\ref{sec:bipartite} shows the decidability of synchronizability
on ring topologies.
Section~\ref{sec:extensions} concludes with
various discussions, including
counter-examples about the mailbox semantics.

\subsection*{Related Work}
The analysis of systems of communicating finite state machines has always been a very active topic
of research. Systems with channel-recognizable (aka QDD~\cite{BoigelotG99} representable)
reachability sets are known to enjoy a decidable reachability problem~\cite{Pachl87}.
Heussner~\emph{et al} developed a CEGAR approach based on
regular model-checking~\cite{HeussnerGS12}. Classifications of communication
topologies according to the decidability of the reachability problems are known
for FIFO, FIFO+lossy, and FIFO+bag communications~\cite{ChambartS08,ClementeHS14}.
In~\cite{TorreMP08,Heussneretal12},
the bounded context-switch reachability problem for
communicating machines extended with local stacks modeling recursive function calls is shown decidable
under various assumptions.
Session types dialects have been introduced for systems of communicating finite state
machines~\cite{DenielouY12},
and were shown to enforce various desirable properties.

Several notions similar to the one of synchronizability have also been
studied in different context.
\emph{Slack elasticity} seems to be the most general name given to
a the property that a given distributed system
with asynchronous communications ``behaves the same''
whatever the slack of communications is. This property has been
studied in hardware design~\cite{Manohar1998},
with the goal of ensuring that some code transformations
are semantic-preserving, in high performance computing,
for ensuring
the absence of deadlocks and other bugs in MPI programs~\cite{Siegel05,Vakkalanka2010}, but also for communicating finite state machines, like in this work,
with a slightly different way of comparing the behaviours of the
system at different buffer bounds. Genest~\emph{et al} introduced the
notion of existentially bounded systems of communicating finite state machines,
that is defined on top of Mazurkiewicz traces, aka message sequence charts in
the context of communicating finite state machines~\cite{GenestKM06}.
Finally, a notion similar to the one of
existentially bounded systems has been recently introduced and
christened ``$k$-synchronous systems''~\cite{Bouajjani-el-al-CAV-2018}.
Existential boundedness, $k$-synchronous systems, and synchronizability are further compared in Section~\ref{sec:existentially-bounded}.
\section{Preliminaries\label{sec:prel}}
\subsection{Messages and topologies}
A \emph{message set}
$\messages$ is a tuple $\tuple{\Sigma_M,\nbpeers,\msgsrc,\msgdst}$ where
$\Sigma_M$ is a finite set of letters (more often called messages),
$\nbpeers\geq 1$ and $\msgsrc,\msgdst$ are functions that associate to every
letter $a\in\Sigma$ naturals $\msgsrc(a)\neq\msgdst(a)\in\{1,\dots,\nbpeers\}$.
We often write $\aij$ for a message $a$ such that $\msgsrc(a)=i$
and $\msgdst(a)=j$; we often identify $M$ and $\Sigma_M$ and write for instance
$M = \{\tagmsg{a_1}{i_1}{j_1},\tagmsg{a_2}{i_2}{j_2},\dots\}$ instead of
$\Sigma_M=\dots$, or $w\in M^*$ instead of
$w\in \Sigma_M^*$. The communication topology associated to $M$ is the graph
$G_M$ with vertices $\{1,\dots,\nbpeers\}$ and with an edge from $i$ to $j$ if there is a message
$a\in\Sigma_M$ such that $\msgsrc(a)=i$ and $\msgdst(a)=j$.
$G_M$ is an \emph{oriented ring} if the
set of edges of $G_M$ is $\{(i,j)\mid i+1=j\mbox{ mod } \nbpeers\}$.

\subsection{Traces}
An \emph{action} $\action$ over $\messages$ is either a send action $!a$ or a
receive action $?a$, with $a\in\Sigma_M$.
The peer $\peerofaction{\action}$ of action $\action$
is defined as $\peerofaction{!a}=\msgsrc(a)$ and $\peerofaction{?a}=\msgdst(a)$.
We write $\actionsofpeer{i}{\messages}$ for the set of actions of peer $i$ and
$\actions{\messages}$ for the set of all actions over $\messages$.
An $M$-\emph{trace} $\trace$ is a
finite (possibly empty) sequence of actions. We write $\actions{\messages}^*$ for the set of
$M$-traces, $\emptytrace$ for the empty $M$-trace, and $\trace_1\cdot{}\trace_2$ for the
concatenation of two $M$-traces. We sometimes write $!?a$ for
$!a\cdot{}?a$.
An $M$-trace $\trace$ is a prefix of
$\tracebis$, $\trace\leqpref\tracebis$ if there is $\traceter$ such that
$\tracebis=\trace\cdot{}\traceter$.
The prefix closure $\prefixclosure S$ of a set of $M$-traces $S$ is the set
$\{\trace\in\actions{\messages}^*\mid\mbox{there is }\tracebis\in S\mbox{ such that }\trace\leqpref\tracebis\}$.

For an $M$-trace $\trace$ and peer ids $i,j\in\{1,\dots,\nbpeers\}$ we write

\begin{itemize}

\item $\sendprojection{\trace}$ (resp. $\recvprojection{\trace}$)
for the sequence of messages sent (resp. received) during $\trace$, \emph{i.e.}
$\sendprojection{!a}=a$, $\sendprojection{?a}=\emptytrace$, and $\sendprojection{\trace_1\cdot{}\trace_2}=\sendprojection{\trace_1}\cdot{}\sendprojection{\trace_2}$ (resp. $\recvprojection{!a}=\emptytrace$,
$\recvprojection{?a}=a$, and $\recvprojection{\trace_1\cdot{}\trace_2}=\recvprojection{\trace_1}\cdot{}\recvprojection{\trace_2}$).

\item $\machineprojection{i}{\trace}$ for the $M$-trace of actions $\action$ in $\trace$
such that $\peerofaction{\action}=i$.

\item $\channelprojection{i\to j}{\trace}$ for the
$M$-trace of actions $\action$ in $\trace$ such that
$\action\in\{!a,?a\}$ for some $a\in M$ with $\msgsrc(a)=i$ and $\msgdst(a)=j$.

\item $\bufferon{i\to j}{\trace}$ for the word $w\in M^*$, if it exists,
such that
$
\sendprojection{\channelprojection{i\to j}{\trace}}=
\recvprojection{\channelprojection{i\to j}{\trace}}\cdot{}w$.

\end{itemize}

\begin{exa}
Consider $M=\tuple{\{a,b\},2,\msgsrc,\msgdst}$ with $\msgsrc(a)=\msgdst(b)=1$ and
$\msgsrc(b)=\msgdst(a)=2$, and let $\trace=!a?b$. Then $\sendprojection{\trace}=a$,
$\machineprojection{1}{\trace}=\trace$,
and $\bufferon{1\to 2}{\trace}=a$.
\end{exa}

An $M$-trace $\trace$ is \emph{FIFO} (resp. a \emph{$k$-bounded FIFO}, for $k\geq 1$)
if for all $i,j\in\{1,\dots,\nbpeers\}$,
for all prefixes $\trace'$ of $\trace$, $\bufferon{i\to j}{\trace'}$ is defined
(resp. defined and of length at most $k$); in other words,
$\trace$ is FIFO if for every prefix $\trace'$ of $\trace$, for all $i\neq j$,
the sequence of messages received from $i$ by $j$ in $\trace' $
is a prefix of the sequence of message sent from $i$ to $j$ in $\trace'$.
Intuitively, an $M$-trace is \emph{FIFO} if it is an execution of
a machine that manipulates FIFO queues, with one queue
per pair of peers.

An $M$-trace is \emph{synchronous}
if it is of the form $!?a_1\cdot{}!?a_2\cdots{}!?a_{\ell}$ for
some $\ell\geq 0$ and $a_1,\dots,a_{\ell}\in M$.
In particular, a synchronous
$M$-trace is a $1$-bounded FIFO $M$-trace (but the converse is false).
An $M$-trace $\trace$ is \emph{stable} if $\bufferon{i\to j}{\trace}=\emptybuffer$ for all
$i\neq j\in\{1,\dots,\nbpeers\}$.

Two $M$-traces $\trace,\tracebis$ are \emph{causal-equivalent} $\trace\causalequiv\tracebis$
if
\begin{itemize}
\item $\trace,\tracebis$ are FIFO, and
\item for all $i\in\{1,\dots,\nbpeers\}$, $\machineprojection{i}{\trace}=
\machineprojection{i}{\tracebis}$.
\end{itemize}
Intuitively, $\trace\causalequiv\tracebis$ if
$\trace$ is obtained from $\tracebis$ by iteratively commuting adjacent actions that
are not from the same peer and do not form a ``matching send/receive pair'' (this is why $\trace,\tracebis$ are deemed to be FIFO).
The relation $\causalequiv$ is a congruence with respect to concatenation.

\medskip
\subsection{Peers, systems, configurations}
A system (of communicating machines) over a message set $\messages$ is a tuple
$\system=\systemlist{\peer_1,\dots,\peer_{\nbpeers}}$ where for all
$i\in\{1,\dots,\nbpeers\}$, the peer
$\peer_i$ is a finite state automaton $\tuple{\states_i,\state_{0,i},\Delta_i}$
over the alphabet $\actionsofpeer{i}{\messages}$ and
with (implicitly) $\states_i$ as the
set of accepting states.
We write $L(\peer_i)$ for the set of $M$-traces that label a path in $\peer_i$
starting at the initial state $\state_{0,i}$.

Let the system $\system$ be fixed. A \emph{configuration} $\configuration$ of
$\system$ is a tuple
$(\state_1,\dots,\state_{\nbpeers},w_{1,2},\dots,w_{\nbpeers-1,\nbpeers})$
where $\state_i$ is a state of $\peer_i$ and for all $i\neq j$,
$w_{i,j}\in \messages^*$ is the content of channel $i\to j$.
A configuration is \emph{stable}
if $w_{i,j}=\emptybuffer$ for all $i,j \in\{1,\dots,\nbpeers\}$ with $i \neq j$.

Let $\configuration=(\state_1,\dots,\state_{\nbpeers},w_{1,2},\dots,w_{\nbpeers-1,\nbpeers})$,
$\configuration'=(\state_1',\dots,\state_{\nbpeers}',w_{1,2}',\dots,w_{\nbpeers-1,\nbpeers}')$ and
$\amessage\in\messages$ with $\msgsrc(\amessage)=i$
and $\msgdst(\amessage)=j$.
We write $\configuration\tr{!\amessage}_{\system}\configuration'$
(resp. $\configuration\tr{?\amessage}_{\system}\configuration'$)
if $(\state_i,!\amessage,\state_i')\in\transitions_i$
(resp. $(\state_j,?\amessage,\state_j')\in\transitions_j$),
$w_{i,j}'=w_{i,j}\cdot{}\amessage$
(resp. $w_{i,j}=\amessage\cdot{}w_{i,j}'$)
and for all $k,\ell$ with $k\neq i$ (resp. with $k\neq j$),
$\state_{k}=\state_{k}'$ and $w_{k,\ell}'=w_{k,\ell}$ (resp. $w_{\ell,k}'=w_{\ell,k}$).
If $\trace=\action_1\cdot{}\action_2\cdots{}\action_n$,
we write $\tr{\trace}_{\system}$ for 
$\tr{\action_1}_{\system}\tr{\action_2}_{\system}\dots\tr{\action_n}_{\system}$.
We often write $\tr{\trace}$ instead of $\tr{\trace}_{\system}$
when $\system$ is clear from the context.
The \emph{initial configuration} of $\system$ is the stable configuration $\configuration_0=(\state_{0,1},\dots,\state_{0,\nbpeers},\emptybuffer,\dots,\emptybuffer)$.
An $M$-trace $\trace$ is a trace of system $\system$ if
there is $\configuration$ such that
$\configuration_0\tr{\trace}\configuration$.
Equivalently, $\trace$ is a trace of $\system$ if
\begin{itemize}
\item it is a FIFO
trace, and
\item for all $i\in\{1,\dots,\nbpeers\}$,
$\machineprojection{i}{\trace}\in L(\peer_i)$.
\end{itemize}
For $k\geq 1$, we write $\traces{k}{\system}$ for
the set of $k$-bounded traces of $\system$, $\traces{0}{\system}$ for the set of synchronous
traces of $\system$, and
$\traces{}{\system}$ for $\bigcup_{k\geq 0}\traces{k}{\system}$.

\begin{exa}\label{ex:counter-example-peer-to-peer}
Consider the message set
$\messages=\{\aonetwo,\bonethree,\cthreetwo,\dtwoone\}$
and the system $\system=\systemlist{\peer_1,\peer_2,\peer_3}$ where $\peer_1,\peer_2,\peer_3$ are as depicted in Fig.~\ref{fig:counter-example}.
$$
\begin{array}{ll}
L(\peer_1)= & \prefixclosure\{!\aonetwo\cdot{} !\aonetwo\cdot{} !\bonethree\}\\
L(\peer_2)= & \prefixclosure\{?\aonetwo\cdot{} ?\aonetwo\cdot{} ?\cthreetwo ~,~ ?\cthreetwo\cdot{}!\dtwoone\}\\
L(\peer_3)= & \prefixclosure\{?\bonethree\cdot{} !\cthreetwo\}.
\end{array}
$$

An example of a stable trace  is
$!\tagmsg{a}{1}{2}\cdot{}!\tagmsg{a}{1}{2}\cdot{}!?\tagmsg{b}{1}{3}\cdot{}!\tagmsg{c}{3}{2}\cdot{}?\tagmsg{a}{1}{2}\cdot{}?\tagmsg{a}{1}{2}\cdot{}?\tagmsg{c}{3}{2}$.
Let $\trace=!\aonetwo\cdot{}!\aonetwo\cdot{}!?\bonethree\cdot{}!?\cthreetwo\cdot{}!\dtwoone$.
Then $\trace\in\traces{2}{\system}$
is a $2$-bounded trace of the system $\system$, and
$\configuration_0\tr{\trace}(\state_{3,1},\state_{5,2},\state_{2,3},\tagmsg{a}{1}{2}\tagmsg{a}{1}{2},\emptybuffer,\tagmsg{d}{2}{1},\emptybuffer,\emptybuffer,\emptybuffer)$.
\end{exa}
\begin{figure}
\begin{tikzpicture}[shorten >=1pt,=stealth’,initial text={},auto,every state/.style={scale=.7,initial distance={2mm}},scale=2]

\begin{scope}
\node[state,initial] (m1q0) at (0,0) {$\state_{0,1}$};
\node[state] (m1q1) at (1,0) {$\state_{1,1}$};
\node[state] (m1q2) at (2,0) {$\state_{2,1}$};
\node[state] (m1q3) at (3,0) {$\state_{3,1}$};
\node[left of=m1q0] {$\peer_1$};
\draw[->] (m1q0) edge node {$!\tagmsg{a}{1}{2}$} (m1q1);
\draw[->] (m1q1) edge node {$!\tagmsg{a}{1}{2}$} (m1q2);
\draw[->] (m1q2) edge node {$!\tagmsg{b}{1}{3}$} (m1q3);
\end{scope}

\begin{scope}[xshift=4cm]
\node[state,initial] (m2q0) at (0,-1) {$\state_{0,2}$};
\node[state] (m2q1) at (0,0) {$\state_{1,2}$};
\node[state] (m2q2) at (1,0) {$\state_{2,2}$};
\node[state] (m2q3) at (2,0) {$\state_{3,2}$};
\node[state] (m2q4) at (1,-1) {$\state_{4,2}$};
\node[state] (m2q5) at (2,-1) {$\state_{5,2}$};
\node[left of=m2q0] {$\peer_2$};
\draw[->] (m2q0) edge node {$?\tagmsg{a}{1}{2}$} (m2q1);
\draw[->] (m2q1) edge node {$?\tagmsg{a}{1}{2}$} (m2q2);
\draw[->] (m2q2) edge node {$?\tagmsg{c}{3}{2}$} (m2q3);
\draw[->] (m2q0) edge node {$?\tagmsg{c}{3}{2}$} (m2q4);
\draw[->] (m2q4) edge node {$!\tagmsg{d}{2}{1}$} (m2q5);
\end{scope}

\begin{scope}
\node[state,initial] (m3q0) at (0,-1) {$\state_{0,3}$};
\node[state] (m3q1) at (1,-1) {$\state_{1,3}$};
\node[state] (m3q2) at (2,-1) {$\state_{2,3}$};
\node[left of=m3q0] {$\peer_3$};
\draw[->] (m3q0) edge node {$?\tagmsg{b}{1}{3}$} (m3q1);
\draw[->] (m3q1) edge node {$!\tagmsg{c}{3}{2}$} (m3q2);
\end{scope}

\end{tikzpicture}
\caption{\label{fig:counter-example}System of Example~\ref{ex:counter-example-peer-to-peer} and
Theorem~\ref{thm:bugBultan}.}
\end{figure}
Two traces $\trace_1,\trace_2$ are $\system$-\emph{equivalent},
denoted with $\trace_1\systemequiv\trace_2$, if
$\trace_1,\trace_2\in\traces{}{\system}$ and there is $\configuration$
such that $\configuration_0\tr{\trace_i}\configuration$ for both $i=1,2$.
It follows from the definition
of $\causalequiv$ that
if $\trace_1\causalequiv\trace_2$ and $\trace_1,\trace_2\in\traces{}{\system}$,
then $\trace_1\systemequiv\trace_2$.

\medskip
\subsection{Synchronizability.}
Following~\cite{BasuBO12}, we define the observable behaviour of a system as its set
of send traces enriched with their final configurations when they
are stable. Formally, for any $k\geq 0$, we write
$\sendtracesweak{k}{\system}$ and $\sendtraces{k}{\system}$ for the sets
$$
\begin{array}{ll}
\sendtracesweak{k}{\system}= &
\{\sendprojection{\trace}\mid\trace\in\traces{k}{\system}\}
\\
\sendtraces{k}{\system}= & \sendtracesweak{k}{\system}
\cup
\{(\sendprojection{\trace},\configuration)\mid \configuration_0\tr{\trace}\configuration,\configuration~\mbox{stable},\trace\in\traces{k}{\system}\}.
\end{array}
$$
Synchronizability is then defined as the slack elasticity of these observable
behaviours.
\begin{defi}[Synchronizability~\cite{BasuB11,BasuBO12}]
A system $\system$ is \emph{synchronizable} if $\sendtraces{0}{\system}=\sendtraces{}{\system}$.
$\system$ is called \emph{language synchronizable} if
$\sendtracesweak{0}{\system}=\sendtracesweak{}{\system}$.
\end{defi}



For convenience, we also introduce a notion of $k$-synchronizability:
for $k\geq 1$, a system $\system$ is $k$-\emph{synchronizable} if
$\sendtraces{0}{\system}=\sendtraces{k}{\system}$,
and \emph{language} $k$-synchronizable if $\sendtracesweak{0}{\system}=\sendtracesweak{k}{\system}$.
A system is therefore
(language) synchronizable if and only if it is (language)
$k$-synchronizable for all $k\geq 1$.

\begin{thm}\label{thm:bugBultan}
There is a system $\system$ that is $1$-synchronizable, but not
synchronizable.
\end{thm}

\begin{proof}
Consider again the system $\system$ of Example~\ref{ex:counter-example-peer-to-peer}. Let $\configuration_{ijk}:=(\state_{i,1},\state_{j,2},\state_{k,3},\emptybuffer,\dots,\emptybuffer)$.
If the buffers are $1$-bounded, $\peer_1$ must wait that the first
$a$ message has been received before sending the $b$ message. Therefore
$$
\begin{array}{lll}
\sendtracesweak{0}{\system} & = &
\prefixclosure \{\aonetwo\cdot{}\aonetwo\cdot{}\bonethree\cdot{}\cthreetwo\}\\
\sendtracesweak{1}{\system} & = & \sendtracesweak{0}{\system}\\
\end{array}
$$
On the other hands, if the buffers can host two transiting
messages, it becomes possible for $\peer_1$ to send $b$ before the first $a$ is
received by $\peer_2$, so it becomes possible for $\peer_3$ to receive $b$
and send $c$, and finally $\peer_2$ may decide to receive $c$ before receiving
any $a$ message. Consequently,
$$
\begin{array}{lll}
\sendtracesweak{2}{\system} & = &
\prefixclosure \{\aonetwo\cdot{}\aonetwo\cdot{}\bonethree\cdot{}\cthreetwo\cdot{}\dtwoone\}
\\
\sendtraces{k}{\system} & = & \sendtracesweak{k}{\system}\cup \mathsf{Stab}
\qquad \mbox{for all }k\geq 0
\end{array}
$$
where $\mathsf{Stab}=
\{
(\emptytrace,\configuration_0),(\tagmsg{a}{1}{2},\configuration_{110}),
(\tagmsg{a}{1}{2}\cdot{}\tagmsg{a}{1}{2},\configuration_{220}),
(\tagmsg{a}{1}{2}\cdot{}\tagmsg{a}{1}{2}\cdot{}\tagmsg{b}{1}{3},\configuration_{321}),
(\tagmsg{a}{1}{2}\cdot{}\tagmsg{a}{1}{2}\cdot{}\tagmsg{b}{1}{3}\cdot{}\tagmsg{c}{3}{2},\configuration_{332})\}
$.
\end{proof}

This example contradicts\footnote{see also the discussion in Section~\ref{sec:error-explaned}} Theorem~4 in~\cite{BasuB16}, which stated
that
$\sendtracesweak{0}{\system}=\sendtracesweak{1}{\system}$ implies
$\sendtracesweak{0}{\system}=\sendtracesweak{}{\system}$.
This also shows that the decidability of synchronizability
for peer-to-peer communications is open despite the claim in~\cite{BasuB16}.
The next section closes this question.

\begin{rem}
In Section~\ref{sec:extensions}, we give a counter-example that addresses
communications with mailboxes, \emph{i.e.} the first communication model considered in all works
about synchronizability, and we list several other published theorems that our counter-example contradicts.
\end{rem}
\section{Undecidability of Synchronizability\label{sec:undecidability}}

In this section, we show the undecidability of synchronizability for systems
with \emph{at least three} peers. Although the reachability problem is undecidable for two peers, we
cannot establish the undecidability of the synchronizability of a system with two peers. The reasons for that are twofolds.

First, synchronizability only deals with messages that are sent and received, which is orthogonal to reachability. We therefore rely on the undecidability of the message reception problem: given a FIFO automaton $\A$ (i.e. an automaton
that can both enqueue and dequeue messages in a unique
channel) and a special message $\specialmsg$, decide whether there exists a trace of $\A$ that contains $?\specialmsg$.

Second, synchronizability constrains a lot the communications. In particular, when an automaton must be in a mixed state (ready to send and receive), it imposes some commutativity of the two actions (see next section), and as a consequence, a synchronizable system with two peers cannot simulate a FIFO automaton.
A third peer is necessary to get rid of all the constraints imposed by synchronizability.

To sum up, we reduce the message reception problem
on a FIFO automaton $\A$ to the synchronizability of a system with three peers:
we construct a system $\system_{\A,m}''$ such that the synchronizability of
$\system_{\A,m}''$ is equivalent to the non-reception
of the special message $\specialmsg$ in $\A$.

\subsection{An Undecidable Problem on FIFO automata}
A \emph{FIFO automaton} is a finite state automaton $\A=\tuple{\states,\actions{\Sigma},\transitions,\state_0}$ over an alphabet of the form
$\actions{\Sigma}$ for some finite set of letters $\Sigma$ with all states
being accepting states.
A FIFO automaton can be thought as a system with only one peer, with the
difference
that, according to our definition of systems, a peer can only
send messages to peers different from itself,
whereas a FIFO automaton enqueues and dequeues letters in a unique FIFO
queue, and thus, in a sense, ``communicates with itself''.
All notions we introduced for systems are obviously extended to FIFO automata.
In particular, a configuration of $\A$ is a tuple
$\configuration=(\state,w)\in\states\times\Sigma^*$, it is stable if
$w=\emptybuffer$, and the transition relation
$\configuration\tr{\trace}\configuration'$ is defined exactly the same way as for
systems.Let us now state the problem that we will consider

\begin{defi}[Message reception problem]
The \emph{message reception problem} is the following decision problem
\begin{description}
\item[Input] a FIFO automaton $\A=\tuple{\states,\actions{\Sigma},\transitions,\state_0}$ and a distinguished message $\specialmsg\in \Sigma$.
\item[Question] is there a trace $\trace$ such that $\trace\cdot{}?m\in\traces{}{\A}$ ?
\end{description}
\end{defi}

\begin{rem}
 A similar but different problem to the message reception problem is the \emph{executable reception problem }in  \cite{BrandZ} which consists to decide for a given control-state $q$ and a message $m$ such that $q \tr{?m}$, whether there exists a reachable configuration $(q,mw)$ with $w \in \Sigma^*$ : this problem is proved undecidable for systems of $2$-CFSMs in the (non available) associated technical report \cite{BrandZ81}. The proof reduces the halting problem to the executable reception problem by using the two unidirectional FIFO channels to simulate the tape of a Turing machine. We present another proof technique (using the undecidability of the existence of a tiling ~\cite{Lewis-Papadimitriou}) for another model, the FIFO automata model and another property. Then we will simulate any FIFO automaton $\A$ by an associated particular system of $3$-CFSMs $\system_{\A}$ that will have the property to be synchronizable iff a message $m$ never appears in a trace of $\A$. We dont try to simulate $\A$ by a system of $2$-CFSMs since for systems of $2$-CFSMs, synchronizability is decidable and message reception problem is undecidable.
 \end{rem}

\begin{lem}\label{lem:undecidability-FIFO-automaton}
The message reception problem is undecidable.
\end{lem}

\begin{proof}
We consider the problem of
the existence of a finite, but arbitrarily large,
tiling for a given set of Wang tiles
and a pair of initial and final tiles. More precisely,
consider a tuple $\tilingset = \tuple{\tiles,\atile_0,\atile_F,H,V}$ where
\begin{itemize}
\item
$\tiles$ is a finite set of tiles,
\item
$\atile_0,\atile_F\in \tiles$ are respectively the initial tile and the final tile, and
\item
$H,V\subseteq \tiles\times \tiles$
are horizontal and vertical compatibility relations.
\end{itemize}
Without loss of generality, we assume that there is a ``padding tile''
$\blank$ such that $(\atile,\blank)\in H\cap V$ for all $\atile\in\tiles$.
For a natural $n\geq 1$,
a $n$-tiling is a function $f:\Nat\times\{1,\dots,n\} \to \tiles$ such that
\begin{itemize}
\item $f(0,1)=\atile_0$,
\item there are $(i_F,j_F)\in\Nat\times\{1,\dots,n\}$ such
that $f(i_F,j_F)=\atile_F$,
\item $(f(i,j),f(i,j+1))\in H$ for all $(i,j)\in\Nat\times\{1,\dots,n-1\}$, and
\item $(f(i,j),f(i+1,j))\in V$ for all $(i,j)\in\Nat\times\{1,\dots,n\}$.
\end{itemize}
The problem of deciding, given a tuple $\tilingset=\tuple{\tiles,\atile_0,\atile_F,H,V}$,
whether there is some $n\geq 1$ for which there exists
a $n$-tiling, is undecidable \cite{Lewis-Papadimitriou}, intuitively because
it is equivalent to the halting problem for a Turing machine working with a
half-infinite ribbon.

In the remainder, we explain how this tiling problem can be reduced
to the message reception problem.
Intuitively, we construct a FIFO automaton that outputs the first row of the tiling, storing it into the queue, and guessing at the same time the width $n$ of the
tiling. Then, for all next row $i+1$, the automaton outputs the row tile after tile, popping a tile of row $i$ in the queue
in between so as to check that each tile of row $i+1$ vertically
coincides with the corresponding tile of row $i$.

More precisely, let $\tilingset=\tuple{\tiles,\atile_0,\atile_F,H,V}$ be fixed.
We define the FIFO automaton
$\A_\tilingset=\tuple{\states,\Sigma,\transitions,q_0}$ with
$\states=\{\statefirstline{\atile},\statebelow{\atile},\stateleft{\atile},\stateleftbelow{\atile}{\atile'}\mid \atile\in \tiles,\atile'\in\tiles\cup\{\cutmsg\}\}\cup\{\state_0,\state_1\}$,
$\Sigma=\tiles\cup\{\cutmsg\}$, and
 $\transitions\subseteq Q\times \actions{\Sigma}\times Q$, with
$$
\begin{array}{lll}
\transitions & = & \{(\state_0,!\atile_0,\statefirstline{\atile_0})\}
\cup
\{(\statefirstline{\atile},!{\atile'},\statefirstline{\atile'})\mid (\atile,\atile')\in H\}
\cup
\{(\statefirstline{\atile},!\cutmsg,\state_1)\mid \atile\in \tiles\}
\\ & \cup &
\{(\state_1,?\atile,\statebelow{\atile})\mid \atile\in \tiles)\}
\cup
\{(\statebelow{\atile},!{\atile'},\stateleft{\atile'})\mid (\atile,\atile')\in V\}
\\ & \cup &
\{(\stateleft{\atile},?{\atile'},\stateleftbelow{\atile}{\atile'})\mid \atile\in\tiles,\atile'\in \tiles\cup\{\cutmsg\}\}
\\ & \cup &
\{(\stateleftbelow{\atile}{\atile'},!\atile'',\stateleft{\atile''})\mid (\atile,\atile'')\in H\mbox{ and } (\atile',\atile'')\in V\}
\\ & \cup &
\{(\stateleftbelow{\atile}{\cutmsg},!\cutmsg,\state_{1})\mid \atile\in\tiles\}
\end{array}
$$
Note that the automaton moves to state $\statebelow{\atile}$ after it has popped the first tile $\atile$ of row $i$ (it needs to remember it), then moves to state $\stateleft{\atile'}$ after it has decided to put a tile $\atile'$ on row $i+1$ above tile $\atile$ (it only needs to remember $\atile'$), then moves to state $\stateleftbelow{\atile'}{\atile''}$ after it has popped the second tile $\atile''$ of row $i$, and so on.
Therefore, any execution of $\A_{\tilingset}$ is of the form
$$
!\atile_{1,1}\cdot{}!\atile_{1,2}\cdots !\atile_{1,n}\cdot{}!\cutmsg\cdot{}?\atile_{1,1}\cdot{}!\atile_{2,1}\cdot{}?\atile_{1,2}\cdot{} !\atile_{2,2}\cdots !\atile_{2,n}\cdot{} ?\cutmsg\cdot{}!\cutmsg\cdot{}?\atile_{2,1}\cdot{}!\atile_{3,1}\cdots
$$
where $\atile_{1,1}=\atile_0$, $(\atile_{i,j},\atile_{i+1,j})\in V$ and $(\atile_{i,j},\atile_{i,j+1})\in H$.
The following two are thus equivalent:
\begin{itemize}
\item
there is $n\geq 1$ such that $\tilingset$ admits a $n$-tiling
\item
there is a trace $\trace\in\traces{}{\A}$ that contains $?\atile_F$.
\qedhere
\end{itemize}
\end{proof}

Note that, from this result, we can easily deduce the undecidability of the reachability problem for a system consisting of two machines, a sender and a receiver, a FIFO channel between them, and an extra channel with rendez-vous synchronization. Indeed, such a system may simulate a FIFO automaton: the sender does exactly the same as the FIFO automaton, except that for reception it uses a rendez-vous synchronization to ask the receiver for performing a reception, and waits for an acknowledgment that this reception has indeed been done. In the following, we will exploit this idea, although not with two machines and a rendez-vous channel, but with three machines and FIFO channels only.

\subsection{A System that Simulates a FIFO Automaton}
We are now going to define a communicating system that
simulates a FIFO automaton $\A$. This system, that we will write
$\system_{\A}$, will later be completed so as to reduce the message reception
problem to synchronizability.
The system $\system_{\A}$ consists of three peers
$\peer_1$, $\peer_2$ and $\peer_3$ that
are connected through the following topology.
\begin{center}
\begin{tikzpicture}[baseline=-.6cm]


\begin{scope} 
\node[rectangle,draw] (P1) at (0,0) {$\peer_1$};
\node[rectangle,draw] (P2) at (2,0) {$\peer_2$};
\node[rectangle,draw] (P3) at (1,-1) {$\peer_3$};
\draw[->] (0.6,0) --++ (.8,0);
\draw[->] ($(P1.south east)+(.1,.1)$) -- ($(P3.north west)+(.2,.1)$);
\draw[<-] ($(P1.south east)+(-.2,-.1)$) -- ($(P3.north west)+(-.1,-.1)$);
\draw[->] ($(P2.south west)+(.2,-.1)$) -- ($(P3.north east)+(.1,.-.1)$);
\draw[<-] ($(P2.south west)+(-.1,.1)$) -- ($(P3.north east)+(-.2,.1)$);

\end{scope}

\end{tikzpicture}
\end{center}
Intuitively, we want $\peer_1$ to mimick $\A$'s decisions and the channel $1\to 2$ to mimick $\A$'s queue. When $\A$ would
enqueue a letter $a$ , peer $1$ sends $\tagmsg{a}{1}{2}$
to peer $2$, and when $\A$ would dequeue a letter $a$,
peer $\peer_1$ sends to peer $\peer_2$ via peer $\peer_3$ the order to dequeue $a$, and waits for
the acknowledgement that the order has been correcly executed. So the only
role of $\peer_3$ is to enable a second channel between $\peer_1$
and $\peer_2$ for ``rendez-vous communications''.

Let us now define these peers and $\system_{\A}$ more formaly. Let
$\A=\tuple{\states_{\A},\actions{\Sigma},\transitions_{\A},\state_0}$ a FIFO
automaton be fixed.
Let $\messages$ be such that all messages
of $\Sigma$ can be exchanged among all peers in all directions but $2\to 1$, \emph{i.e.}
$$
\messages = \{\tagmsg{a}{1}{2},\tagmsg{a}{1}{3},\tagmsg{a}{2}{3},\tagmsg{a}{3}{1},\tagmsg{a}{3}{2}\mid a\in\Sigma\}
$$
Peer $\peer_1$ is obtained by replacing every $!a$ transition of $\A$ with
a $!\tagmsg{a}{1}{2}$ transition, and every $?a$ transition with the
sequence of transitions $!\tagmsg{a}{1}{3}?\tagmsg{a}{3}{1}$.
Formally, $\peer_1=\tuple{\states_1,q_{0,1},\transitions_1}$ is defined by
$\states_{1} =\states_{\A}\uplus\{\state_{\atrans}\mid \atrans\in \transitions_{\A}\}$ and
$$
\begin{array}{lcl}
\transitions_{1} & = &
\{(\state,!\tagmsg{a}{1}{2},\state')\mid (\state,!a,\state')\in\transitions_{\A}\}
\\
& \cup &
\{(\state,!\tagmsg{a}{1}{3},\state_{\atrans}),(\state_{\atrans},?\tagmsg{a}{3}{1},
\state')\mid \atrans=(\state,?a,\state')\in\transitions_{\A}\}.
\end{array}
$$
The \emph{modus operandi} of $\peer_2$ and $\peer_3$ is rather simple: $\peer_3$
propagates all messages it
receives, and $\peer_2$ executes all orders it receives and sends back an
acknowledgement when
this is done. So both $\peer_2$ and $\peer_3$ operate with an ``infinite loop''.
For some technical reason that will be later explained,
we need to make sure that $\peer_2$ never comes back to its initial state. To do
so, we simply ``unroll the loop'' once in $\peer_2$.
Formally,
let $\peer_2=\tuple{\states_2,q_{0,2},\transitions_2}$
and $\peer_3=\tuple{\states_3,q_{0,3},\transitions_3}$ be defined by
$$
\begin{array}{ll}
\begin{array}{l@{~=~}l}
\states_{2} &  \{\state_{0,2},\state_{1,2}\}\cup\{\state_{a,1},\state_{a,2}\mid a\in \Sigma\}\\
\end{array}
&
\begin{array}{l@{~=~}l}
\states_{3} &  \{\state_{0,3}\}\cup\{\state_{a,1},\state_{a,2},\state_{a,3}\mid a\in \Sigma\}
\end{array}
\\
\multicolumn{2}{l}{
\begin{array}{l@{~=~}l}
\transitions_2 &
\{(\state_{0,2},?\tagmsg{a}{3}{2},\state_{a,1}),(\state_{1,2},?\tagmsg{a}{3}{2},\state_{a,1}),(\state_{a,1},?\tagmsg{a}{1}{2},\state_{a,2}),(\state_{a,2},!\tagmsg{a}{2}{3},\state_{1,2})\mid a\in \Sigma\}
\\
\transitions_{3} &
\{(\state_{0,3},?\tagmsg{a}{1}{3},\state_{a,1}),(\state_{a,1},!\tagmsg{a}{3}{2},\state_{a,2}),(\state_{a,2},?\tagmsg{a}{2}{3},\state_{a,3}),(\state_{a,3},!\tagmsg{a}{3}{1},\state_{0,3})\mid a\in\Sigma\}
\end{array}}
\end{array}
$$

\begin{exa}\label{ex:fifo-automaton-encoding}
Consider $\Sigma=\{a,\specialmsg\}$ and
the FIFO automaton $\A=\tuple{\{\state_0,\state_1\},\actions{\Sigma},\transitions,\state_0}$
with transition relation
$\transitions_{\A}=\{(\state_0,!a,\state_0),(\state_0,!\specialmsg,\state_1),(\state_1,?a,\state_0),(\state_1,?\specialmsg,\state_0)\}$.
Then $\A$ and the peers $\peer_1,\peer_2,\peer_3$ are depicted in Fig.~\ref{fig:fifo-automaton-encoding}.
\end{exa}

\begin{figure}
\begin{tikzpicture}[shorten >=1pt,=stealth’,initial text={},auto,every state/.style={scale=.3,initial distance={2mm}}]

\begin{scope}[xshift=3cm]
\draw[dashed] (-1,1.5) rectangle (1.7,-1.5);
\node at (0.5,-2) {$\A$};
\node[state,initial] (aq0) at (0,0) {};
\node[state] (aq1) at (1,0) {};
\draw[->] (aq0) edge [loop above] node[above] {$!a$} (aq0);
\draw[->] (aq0) edge [bend left] node [above] {$!\specialmsg$} (aq1);
\draw[->] (aq1) edge [bend left] node {$?a,?\specialmsg$} (aq0);
\end{scope}

\begin{scope}[xshift=8cm]
\draw[dashed] (-1,2) rectangle (6,-2);
\node at (2.5,-2.5) {$\peer_1$};
\node[state,initial] (p1q0) at (0,0) {};
\node[state] (p1q1) at (5,0) {};
\node[state] (p1q2) at (2.5,-0.5) {};
\node[state] (p1q3) at (2.5,-1.5) {};
\draw[->] (p1q0) edge [loop above] node [above] {$!\tagmsg{a}{1}{2}$} (p1q0);
\draw[->] (p1q0) edge [bend left] node[above] {$!\tagmsg{\specialmsg}{1}{2}$} (p1q1);
\draw[->] (p1q1) edge[bend left=10] node [above] {$!\tagmsg{a}{1}{3}$} (p1q2.east);
\draw[->] (p1q2) edge[bend left=10] node [above right] {$?\tagmsg{a}{3}{1}$} (p1q0);
\draw[->] (p1q1) edge[bend left] node [below right] {$!\tagmsg{\specialmsg}{1}{3}$} (p1q3);
\draw[->] (p1q3) edge[bend left] node [below left] {$?\tagmsg{\specialmsg}{3}{1}$} (p1q0);

\end{scope}

\begin{scope}[yshift=-5cm,xshift=3cm,xscale=2]
\node at (1,-3) {$\peer_2$};
\draw[dashed] (-1,-2.5) rectangle (3,2.5);
\node[state,initial,initial left] (p2q0) at (0,0) {};
\node[state](p2q0') at (1,0) {};
\node[state] (p2q1) at (1,-2) {};
\node[state] (p2q2) at (2.5,-1) {};
\node[state] (p2q1') at (1,2) {};
\node[state] (p2q2') at (2.5,1) {};
\draw[->] (p2q0) edge node[below left] {$?\tagmsg{a}{3}{2}$} (p2q1);
\draw[->] (p2q0') edge node[right] {$?\tagmsg{a}{3}{2}$} (p2q1);
\draw[->] (p2q1) edge node[below right] {$?\tagmsg{a}{1}{2}$} (p2q2);
\draw[->] (p2q2) edge node[above right] {$!\tagmsg{a}{2}{3}$} (p2q0');
\draw[->] (p2q0) edge node [above left] {$?\tagmsg{\specialmsg}{3}{2}$} (p2q1');
\draw[->] (p2q0') edge node [right] {$?\tagmsg{\specialmsg}{3}{2}$} (p2q1');
\draw[->] (p2q1') edge node [above right] {$?\tagmsg{\specialmsg}{1}{2}$} (p2q2');
\draw[->] (p2q2') edge node[below right] {$!\tagmsg{\specialmsg}{2}{3}$} (p2q0');
\end{scope}

\begin{scope}[xscale=1.5,xshift=2cm,yshift=-15cm]
\node at (4,-2) {$\peer_3$};
\draw[dashed] (-1,-1.5) rectangle (9,1.5);
\node[state,initial,initial above] (p3q0) at (4,0) {};
\node[state] (p3q1) at (6,.5) {};
\node[state] (p3q2) at (8,0) {};
\node[state] (p3q3) at (6,-.5) {};
\node[state] (p3q1') at (2,.5) {};
\node[state] (p3q2') at (0,0) {};
\node[state] (p3q3') at (2,-.5) {};
\draw[->] (p3q0) edge node[above] {$?\tagmsg{a}{1}{3}$} (p3q1);
\draw[->] (p3q1) edge node[above] {$!\tagmsg{a}{3}{2}$} (p3q2);
\draw[->] (p3q2) edge node[below] {$?\tagmsg{a}{2}{3}$} (p3q3);
\draw[->] (p3q3) edge node[below] {$!\tagmsg{a}{3}{1}$} (p3q0);
\draw[->] (p3q0) edge node[above] {$?\tagmsg{\specialmsg}{1}{3}$} (p3q1');
\draw[->] (p3q1') edge node[near start,above] {$!\tagmsg{\specialmsg}{3}{2}$} (p3q2');
\draw[->] (p3q2') edge node[below] {$?\tagmsg{\specialmsg}{2}{3}$} (p3q3');
\draw[->] (p3q3') edge node[below] {$!\tagmsg{\specialmsg}{3}{1}$} (p3q0);
\end{scope}

\begin{scope}[xshift=4cm,xscale=1.5,yshift=-10cm]
\node at (0.5,-2.5) {$\peer_2'$};
\draw[dashed] (-2,-2) rectangle (3,1.25);
\node[state,initial] (p2'q0) at (-.7,0) {};
\node[state] (p2'q1) at (1,0) {};
\node[state] (p2'q2) at (1,-1) {};
\draw[->] (p2'q0) edge node [above] {$?\tagmsg{a}{1}{2},?\tagmsg{\specialmsg}{1}{2}$} (p2'q1);
\draw[->] (p2'q1) edge [in=0,out=60,loop] (p2'q1);
\node at ($(p2'q1)+(1,.5)$) {$\begin{array}{c}?\tagmsg{a}{1}{2}, ?\tagmsg{\specialmsg}{1}{2}\end{array}$};
\draw[->] (p2'q1) edge [bend left] node [right] {$?\tagmsg{a}{3}{2}$} (p2'q2);
\draw[->] (p2'q2) edge [in=-60,out=0,loop] (p2'q2);
\node at ($(p2'q2)+(1,-.5)$) {$\begin{array}{c}?\tagmsg{a}{1}{2}, ?\tagmsg{\specialmsg}{1}{2}\end{array}$};
\draw[->] (p2'q2) edge [bend left] node [left] {$!\tagmsg{a}{2}{3}$} (p2'q1);
\draw[->] (p2'q0) |- node [pos=.6,below] {$?\tagmsg{a}{3}{2}$} (p2'q2);
\end{scope}

\end{tikzpicture}
\caption{\label{fig:fifo-automaton-encoding}
The FIFO automaton $\A$ of Example~\ref{ex:fifo-automaton-encoding}
and its associated systems $\system_{\A}=\systemlist{\peer_1,\peer_2,\peer_3}$
and $\system_{\A,\specialmsg}'=\systemlist{\peer_1,\peer_2',\peer_3}$.
The sink state $\state_{\bot}$ and the transitions
$\state\tr{?\tagmsg{\specialmsg}{3}{2}}\state_{\bot}$ are omitted in the representation of $\peer_2'$.}
\end{figure}

Let $\system_{\A}=\systemlist{\peer_1,\peer_2,\peer_3}$.
There is a
tight correspondence between the $k$-bounded traces of $\A$, for $k\geq 1$, and
the $k$-bounded traces of $\system_{\A}$: every trace
$\tau\in\traces{k}{\A}$ induces the trace $h(\tau)\in\traces{k}{\system_{\A}}$
where
$h:\actions{\Sigma}\to\actions{\messages}^*$
is the homomorphism from the traces of $\A$ to the traces of $\system_{\A}$
defined by $h(!a)=!\tagmsg{a}{1}{2}$ and
$h(?a)=!?\tagmsg{a}{1}{3}\cdot{}!?\tagmsg{a}{3}{2}\cdot{}?\tagmsg{a}{1}{2}\cdot{}!?\tagmsg{a}{2}{3}\cdot{}!?\tagmsg{a}{3}{1}$. The converse is not
true: there are traces of $\system_{\A}$ that are not prefixes
of a trace $h(\trace)$ for some $\trace\in\traces{k}{\A}$. This happens
when $\peer_1$ sends an order to dequeue $\tagmsg{a}{1}{3}$
that correspond to a transition $?a$ that $\A$ cannot execute. In that case,
the system blocks when $\peer_2$ has to execute the order. To sum up, we obtain the following result.


\begin{lem}\label{lem:fifo-auto-encoding}
For all $k\geq 0$,
$$
\begin{array}{l@{~}l@{~}l}
\traces{k}{\system_{\A}} & = &
\prefixclosure\{h(\trace)\mid\trace\in\traces{k}{\A}\}
\\ & \cup &
\prefixclosure\{h(\trace)\cdot{}!?\tagmsg{a}{1}{3}\cdot{}!?\tagmsg{a}{3}{2}\mid\exists q,b,w,q'\mbox{ s.t. }
\\ && ~~~~~\trace\in\traces{k}{\A},(\state_0,\emptybuffer)\tr{\trace}(\state,bw),(\state,?a,\state')\in\transitions
\mbox{and }b\neq a\}.
\end{array}
$$
\end{lem}

\subsection{A Synchronizable System}

Let us fix a special message $\specialmsg\in\Sigma$.
In this section, we define a system
$\system'_{\A,\specialmsg}=\systemlist{\peer_1,\peer_2',\peer_3}$
where
$\peer_1$ and $\peer_3$ are like in the system $\system_{\A}$,
but $\peer_2'$ is a new peer.
This system will later be combined with $\system_{\A}$ so
as to form the whole system that will be used in the reduction of
the message reception problem to the synchronizability problem.
The main purpose
of $\system'_{\A,\specialmsg}=\systemlist{\peer_1,\peer_2',\peer_3}$
is to be a synchronizable system that will ``make synchronizable''
all traces of $\system_{\A}$ except the ones that contain $\tagmsg{\amessage}{2}{3}$, which are the only ones we want to care about in the reduction.
Our outline for this section is therefore the following: (1) define $\system_{\A,\specialmsg}'$, (2) compute
its synchronous traces, (3) show that they ``make synchronizable''
the asynchronous traces of $\system_{\A}$ where
$!\tagmsg{\specialmsg}{2}{3}$ does not occur,
and (4)
show that it is a synchronizable system.

Let us start with the definition of $\system'_{\A,\specialmsg}=\systemlist{\peer_1,\peer_2',\peer_3}$.
Intuitively, the new peer $\peer_2'$ will always be able to receive any message
from peer $\peer_1$, in particular at the time the message is sent. Moreover,
like $\peer_2$, $\peer_2'$ will also be able to receive
orders to dequeue from peer $\peer_3$,
but instead of executing the order before sending an acknowledgement,
it will ignore the order as follows. If $\peer_2'$ receives the order
to dequeue a message $\tagmsg{a}{1}{2}\neq\tagmsg{\specialmsg}{1}{2}$,
$\peer_2'$ acknowledges $\peer_3$ but does not dequeue $a$ in the $1\to2$ queue.
If the order
was to dequeue $\specialmsg$, $\peer_2'$ blocks in the sink state $\state_{\bot}$ and does not send an acknowledgement to $\peer_3$. As for $\peer_2$,
we ``unroll the loop''
 so as to make sure that it not possible to come back to the initial
state of $\peer_2'$. Formally, $\peer_2'$ is defined as follows.
$$
\begin{array}{lll}
\states_2' & = &
\{\state_{0,2},\state_{0,2}'\}\cup\{\state_{a,1}'\mid a\in\Sigma,a\neq \specialmsg,\}\cup\{\state_{\bot}\}
\\
\transitions_2' & = &
\{(\state_{0,2},?\tagmsg{a}{1}{2},\state_{0,2}'),(\state,?\tagmsg{a}{1}{2},\state)\mid a\in\Sigma,\state\neq\state_{0,2}\}
\\ & \cup &
\{(\state_{0,2},?\tagmsg{a}{3}{2},\state_{a,1}'),(\state_{0,2}',?\tagmsg{a}{3}{2},\state_{a,1}'),(\state_{a,1}',!\tagmsg{a}{2}{3},\state_{0,2}'),\mid a\in\Sigma,a\neq \specialmsg\}
\\ & \cup & \{(\state,?\tagmsg{\specialmsg}{3}{2},\state_{\bot})\mid q\in \states_2'\}
\end{array}
$$

\begin{exa}
For $\Sigma=\{a,\specialmsg\}$, and $\A$ as in Example~\ref{ex:fifo-automaton-encoding}, $\peer_2'$ is depicted in Fig.~\ref{fig:fifo-automaton-encoding}
(omitting the transitions to the sink state $\state_{\bot}$).
\end{exa}

Let us now
compute the set of all synchronous traces of $\system_{\A,\specialmsg}'$.
Observe first that the system $\system_{\A,\specialmsg}'=\systemlist{\peer_1,\peer_2',\peer_3}$ contains
many synchronous traces:
when
$\peer_1$ sends a message $\tagmsg{a}{1}{2}$, it can always do it synchronously,
because $\peer_2'$ is always ready to receive it. When $\peer_1$ sends an order for dequeuing,
the transmission of this order to $\peer_2'$ through $\peer_3$
can be synchronous. If this order is not the order
to dequeue $\tagmsg{\specialmsg}{1}{2}$, then $\peer_2'$ sends the
acknowledgment to $\peer_1$ through $\peer_3$, which can also happen
synchronously.
Note in particular that, unlike in $\system_{\A}$,
peer 1 does not block forever after it has sent
an order $\tagmsg{a}{1}{3}$ in a configuration where the first
message to be dequeued in channel $1\to 2$ is not $a$,
because $\peer_2'$ now acknowledges any order (except for $\specialmsg$).
Therefore
any trace $\trace$ labeling a path in automaton
$\peer_1$ can be lifted to a synchronous trace $\tau'\in\traces{0}{\system_{\A,\specialmsg}'}$
provided $!\tagmsg{\specialmsg}{1}{3}$ does not occur in $\tau$. However, if
$\peer_1$ takes a $!\tagmsg{\specialmsg}{1}{3}$ transition,
it gets
blocked for ever waiting for $\tagmsg{\specialmsg}{3}{1}$.
Therefore, if $!\tagmsg{m}{1}{3}$ occurs in a synchronous trace $\tau$ of
$\system_{\A,\specialmsg}'$, it must be in the last four actions, and this
trace leads to a deadlock configuration in which both 1 and 3 wait for an acknowledgement and 2 is in the sink state.

Let us now formalize further these observations.
Let $L^{\specialmsg}(\A)$ be the set of traces $\trace$
recognized by $\A$ as a finite state automaton (over the alphabet
$\actions{\Sigma}$) such that either
$?\specialmsg$ does not occur in $\trace$, or it occurs only once and it
is the last action of $\trace$.

\begin{exa}
With $\A$ as in
Example~\ref{ex:fifo-automaton-encoding},
$
L^{\specialmsg}(\A)=~\prefixclosure\!
\Big((!a^*\cdot{}!\specialmsg\cdot{}?a)^*\cdot{}!a^*\cdot{}!\specialmsg\cdot{}?\specialmsg\Big).
$
\end{exa}

The next lemma formalizes the observations we did about how synchronous
traces of $\system_{\A,\specialmsg}'$ correspond, up to an homomorphism,
to $L^{\specialmsg}(\A)$, and gives the desired computation of the synchronous traces of $\system_{\A,\specialmsg}'$.

\begin{lem}\label{lem:peer2-variant}
$\traces{0}{\system'_{\A,\specialmsg}}=~\prefixclosure\!
\{h'(\trace)\mid \trace\in L^{\specialmsg}(\A)\}
$,
where $h':\actions{\Sigma}^*\to\actions{\messages}^*$
is the morphism defined by
\begin{itemize}
\item
$h'(!a)=!?\tagmsg{a}{1}{2}$ for all $a\in\Sigma$,
\item
$h'(?a)=!?\tagmsg{a}{1}{3}\cdot{}!?\tagmsg{a}{3}{2}\cdot{}!?\tagmsg{a}{2}{3}\cdot{}!?\tagmsg{a}{3}{1}$ for all $a\neq\specialmsg$, and
\item
$h'(?\specialmsg)=!?\tagmsg{\specialmsg}{1}{3}\cdot{}!?\tagmsg{\specialmsg}{3}{2}$.
\end{itemize}
\end{lem}

As a consequence, we get the following result, which will be later used
to ``make synchronizable'' all traces of $\system_{\A}$ that do not
contain $!\tagmsg{m}{2}{3}$.
\begin{lem}
\label{lem:make-synchronizable}
For all trace $\trace\in\traces{}{\system_{\A}}$ such
that $!\tagmsg{\specialmsg}{2}{3}\not\in\trace$, there is a synchronous
trace $\trace'\in\traces{0}{\system_{\A,\specialmsg}'}$ such that
$\sendprojection{\trace}=\sendprojection{\trace'}$.
\end{lem}
\begin{proof}
Let  $\trace\in\traces{}{\system_{\A}}$ such
that $!\tagmsg{\specialmsg}{2}{3}\not\in\trace$ be fixed.
By Lemma~\ref{lem:fifo-auto-encoding},
there is $\trace_0\in \traces{}{\A}$ such that
$\trace=h(\trace_0)$. By definition of $h$, $?\specialmsg$
does not occur in $\trace_0$.
Let $\trace'=h'(\trace_0)$.
By Lemma~\ref{lem:peer2-variant},
$\trace'\in\traces{0}{\system_{\A,\specialmsg}'}$,
and by definition of $h$ and $h'$, and the fact that $?\specialmsg$
does not occur in $\trace_0$,
$\sendprojection{\trace}=\sendprojection{\trace'}$.
\end{proof}

Let us finally establish the synchronizability of $\system_{\A,\specialmsg}'$.
We consider some arbitrary asynchronous trace $\trace\in\traces{}{\system_{\A,\specialmsg}'}$ that we need to be equivalent, up to receive actions,
to a synchronous trace. Let us reason message by message on $\trace$,
by case analysis on the channel of the message.

\begin{itemize}
\item
If $\peer_1$ sends a message $\tagmsg{a}{1}{2}$ to $\peer_2'$,
it is always possible to make sure
that $\peer_2'$ receives it immediately. Indeed, there are two cases:
if $\tagmsg{a}{1}{2}$ was not received in $\trace$, adding $?\tagmsg{a}{1}{2}$
right after $!\tagmsg{a}{1}{2}$ in $\trace$ yields a valid trace in
$\traces{}{\system_{\A,\specialmsg}'}$, because the transitions
$?\tagmsg{a}{1}{2}$ in $\peer_2'$ do not modify the control state;
similary, if $\tagmsg{a}{1}{2}$ was received in $\trace$ but not
immediately after $!\tagmsg{a}{1}{2}$, it is possible
to move $?\tagmsg{a}{1}{2}$ immediately after
$!\tagmsg{a}{1}{2}$ in $\trace$ while keeepin a valid trace in
$\traces{}{\system_{\A,\specialmsg}'}$, again because the transitions
$?\tagmsg{a}{1}{2}$ in $\peer_2'$ do not modify the control state. In the remainder, we therefore assume that all $!\tagmsg{a}{1}{2}$ in $\trace$
are immediately followed by $?\tagmsg{a}{1}{2}$.
\item
If $\peer_1$ sends a message to $\peer_3$,
it is always possible to make sure
that $\peer_3$ receives it immediately. Indeed, it is always the case that
whenever $\peer_1$ sends a message to $\peer_3$, $\peer_3$ is in its initial state, otherwise $\peer_1$ would be waiting
for an acknowledgment from $\peer_3$, and won't be able to send a message to
$\peer_3$.
\item For the same reason,
if $\peer_2'$ sends a message $\tagmsg{a}{2}{3}$ to $\peer_3$,
it must be the case that
$\peer_3$ is blocked waiting for this message, and we can either move
$?\tagmsg{a}{2}{3}$ right after $!\tagmsg{a}{2}{3}$ or insert it in $\trace$
if it was not there.
\item
For the same reason, if $\peer_3$ sends a message
to $\peer_1$, it is always
possible to make sure that $\peer_1$ receives it immediately.
\item Finally, let us consider the case of $\peer_3$ sending a message $\tagmsg{a}{3}{2}$ to $\peer_2'$.
It must be the case that
$\peer_2'$ is either in its initial state $q_{0,2}$ or in the similar
receiving state $q_{0,2'}$ at the moment of the sending. Indeed, if $\peer_2'$
was in a state $q_{a,1}'$, $\peer_3$ would be blocked waiting for an
acknowledgment from $\peer_2'$, so it would not have been able to send a message to $\peer_2'$. So $\peer_2'$ is either in its initial state $q_{0,2}$ or in the similar
receiving state $q_{0,2'}$ at the moment of the sending $!\tagmsg{a}{3}{2}$.
With the same reasoning, it also holds that the buffer $3\to 2$ was empty
before the sending of $\tagmsg{a}{3}{2}$.
Since there are no send transitions from
$q_{0,2}$ and $q_{0,2'}$, and since we assumed above that all $?\tagmsg{a}{1}{2}$
immediately follow their matching send in $\trace$,
the only possible first action of $\peer_2'$
in $\trace$ after $!\tagmsg{a}{3}{2}$ is $?\tagmsg{a}{3}{2}$.
If this action exists in $\trace$, we can move it right after the sending
of $\peer_3$ up to causal equivalence.
If $?\tagmsg{a}{2}{3}$ does not happen in $\trace$ after this
$!\tagmsg{a}{3}{2}$, it means that no further action of $\peer_2'$ occurs
in $\trace$ after $!\tagmsg{a}{3}{2}$. So we can insert $?\tagmsg{a}{3}{2}$
in $\trace$ right after $!\tagmsg{a}{3}{2}$ while keeping
a valid trace in $\traces{}{\system_{\A,\specialmsg}'}$.
\end{itemize}

In order to sum up what we showed with this case analysis,
let us introduce the homomorphism
$h'':\actions{\messages}^*\to\actions{\messages}^*$
such that
\begin{itemize}
\item $h''(!\tagmsg{a}{1}{2})=!?\tagmsg{a}{1}{2}$,
\item $h''(?\tagmsg{a}{1}{2})=\emptytrace$, and
\item $h''(\action)=\action$ otherwise.
\end{itemize}

For any given $\trace\in\traces{}{\system_{\A,\specialmsg}'}$,
our case analysis shows that
$h''(\trace)\in\traces{0}{\system_{\A,\specialmsg}'}$.
It is also easy to observe that $\sendprojection{\trace}=\sendprojection{h''(\trace)}$.
As a consequence, we get the desired result.

\begin{lem}\label{lem:system-2-synchronizable}
$\system_{\A,\specialmsg}'$ is synchronizable.
\end{lem}

\subsection{Reducing the Message Reception Problem to Synchronizability}

We are now close to reach our initial goal, namely to reduce the message
reception problem to synchronizability.
Let us
consider the system
$\system''_{\A,\specialmsg}=\systemlist{\peer_1,\peer_2\cup\peer_2',\peer_3}$,
where $\peer_2\cup\peer_2'=\tuple{\states_2\cup\states_2',q_{02},\transitions_2\cup\transitions_2'}$ is obtained by merging
the initial state $\state_{0,2}$ of $\peer_2$ and $\peer_2'$.

It is now time to explain why we defined $\peer_2$ and $\peer_2'$
so that it is not possible to come back to the initial state
$\state_{0,2}$. While doing so, we make sure that
any trace of
$\system''_{\A,\specialmsg}$ is either a trace of $\system_{\A}$ or
a trace of $\system_{\A,\specialmsg}'$.
$$
\traces{k}{\system_{\A,\specialmsg}''}=
\traces{k}{\system_{\A}}\cup
\traces{k}{\system_{\A,\specialmsg}'},
$$
In particular,
$$
\sendtracesweak{k}{\system_{\A,\specialmsg}''}=
\sendtracesweak{k}{\system_{\A}}\cup
\sendtracesweak{k}{\system_{\A,\specialmsg}'},
\qquad \mbox{and} \qquad
\sendtraces{k}{\system_{\A,\specialmsg}''}=
\sendtraces{k}{\system_{\A}}\cup
\sendtraces{k}{\system_{\A,\specialmsg}'}.
$$

The next lemma establishes the soundness of the reduction of message reception
to language synchronizability. The reduction to synchronizability will be later
treated.

\begin{lem}
\label{lem:reduction-message-reception-language-synchronizability}
$\system_{\A,\specialmsg}''$ is not language synchronizable
iff there is a trace $\trace\in\traces{}{\A}$ such that
$?\specialmsg$ occurs in $\trace$.
\end{lem}

\begin{proof}
($\Rightarrow$) Assume that $\system_{\A,\specialmsg}''$ is not language
synchronizable
and let us show that there is a trace $\trace\in\traces{}{\A}$ such that
$?\specialmsg$ occurs in $\trace$.

Let us first observe that if $\system_{\A,\specialmsg}''$ is not synchronizable,
then $\sendtracesweak{}{\system_{\A,\specialmsg}''}\setminus\sendtracesweak{0}{\system_{\A,\specialmsg}''}\neq\emptyset$. Since by Lemma~\ref{lem:system-2-synchronizable} $\system_{\A,\specialmsg}'$ is synchronizable,
$\sendtracesweak{}{\system_{\A,\specialmsg}'}\setminus\sendtracesweak{0}{\system_{\A,\specialmsg}'}=\emptyset$. So
$\sendtracesweak{}{\system_{\A}}\setminus\sendtracesweak{0}{\system_{\A,\specialmsg}''}\neq\emptyset$, and in particular
$$
\sendtracesweak{}{\system_{\A}}\setminus\sendtracesweak{0}{\system_{\A,\specialmsg}'}\neq\emptyset.
$$
Let $\trace_1\in\traces{}{\system_A}$ be such that
$\sendprojection{\trace_1}\not\in\sendtracesweak{0}{\system_{\A,\specialmsg}}$.
Then $!\tagmsg{\specialmsg}{2}{3}$ occurs in $\trace_1$.
Indeed, if it was not the case, then by Lemma~\ref{lem:make-synchronizable},
we would have $\sendprojection{\trace_1}\in\sendtracesweak{0}{\system_{\A,\specialmsg}}$. Now, since $!\tagmsg{\specialmsg}{2}{3}$ occurs in $\trace_1$,
by Lemma~\ref{lem:fifo-auto-encoding}, there is a trace
$\trace\in\traces{}{\A}$ such that $h(\trace)=\trace_1$,
and by definition of $h$, $?\specialmsg$ occurs in $\trace$.

($\Leftarrow$) Assume that there is a trace
$\trace\in\traces{}{\A}$ such that $?\specialmsg$ occurs in $\trace$,
and let us show that $\system_{\A,\specialmsg}''$ is not language synchronizable.
By Lemma~\ref{lem:fifo-auto-encoding},
$h(\trace)\in\traces{}{\system_{\A}}$. In order to show
that $\system_{\A,\specialmsg}''$ is not synchronizable (resp. not
language synchronizable), let us show
that $\sendprojection{h(\trace)}\in\sendtracesweak{}{\system_{\A,\specialmsg}''}$
and
$\sendprojection{h(\trace)}\not\in\sendtracesweak{0}{\system_{\A,\specialmsg}''}$.
\begin{itemize}
\item $\sendprojection{h(\trace)}\in\sendtracesweak{}{\system_{\A,\specialmsg}''}$. Indeed, $\sendprojection{h(\trace)}\in\sendtracesweak{}{\system_A}$ because $h(\trace)\in\traces{}{\system_{\A}}$, and
$\sendtracesweak{}{\system_{\A}}\subseteq\sendtracesweak{}{\system_{\A,\specialmsg}''}$.
\item Let us show that
$\sendprojection{h(\trace)}\not\in\sendtracesweak{0}{\system_{\A,\specialmsg}''}$. Since $\sendtracesweak{0}{\system_{\A,\specialmsg}''}=\sendtracesweak{0}{\system_{\A}}\cup\sendtracesweak{0}{\system_{\A,\specialmsg}'}$, we need
to show that
$\sendprojection{h(\trace)}\not\in\sendtracesweak{0}{\system_{\A}}$
and
$\sendprojection{h(\trace)}\not\in\sendtracesweak{0}{\system_{\A,\specialmsg}'}$.
\begin{itemize}
\item $\sendprojection{h(\trace)}\not\in\sendtracesweak{0}{\system_{\A,\specialmsg}'}$. Indeed, since $?\specialmsg$ occurs in $\trace$, $!\tagmsg{\specialmsg}{2}{3}$ occurs in $h(\trace)$, but by definition,
$\peer_2'$, contains no $!\tagmsg{\specialmsg}{2}{3}$ transition, so
a trace that contains $!\tagmsg{\specialmsg}{2}{3}$ can't be trace of
$\system_{\A,\specialmsg}''$, therefore $\sendprojection{h(\trace)}\not\in\sendtracesweak{0}{\system_{\A,\specialmsg}'}$.
\item
$\sendprojection{h(\trace)}\not\in\sendtracesweak{0}{\system_{\A}}$.
Let us assume by absurd that there is a synchronous trace
$\trace'\in\traces{0}{\system_{\A}}$ such that $\sendprojection{\trace'}=\sendprojection{h(\trace)}$.
By Lemma~\ref{lem:fifo-auto-encoding}, there is
$\trace_0\in\traces{0}{\A}$ such that either $\trace'$ is a prefix of
$h(\trace_0)$, or $\trace'$ is a prefix of
$h(\trace_0)\cdot{}!?\tagmsg{a}{1}{3}\cdot{}!?\tagmsg{a}{3}{2}$ for some $a\in\Sigma$.
Since $\trace'$ is a synchronous trace, and by definition of $h$,
$\trace_0$ must only contain receptions, and since $\trace_0$ corresponds to
a trace in $\A$, $\trace_0$ is the empty trace. So $\trace'$ is
a prefix of $!?\tagmsg{a}{1}{3}\cdot{}!?\tagmsg{a}{3}{2}$ for some $a\in\Sigma$,
and $!\tagmsg{a}{2}{3}$ does not occur in $\trace'$.
This contradicts $\sendprojection{\trace'}=\sendprojection{h(\trace)}$ and
the fact that $!\tagmsg{a}{2}{3}$ does occur in $h(\trace)$.
\end{itemize}
\end{itemize}
\end{proof}

Let us now establish the soundness of the reduction of message reception to
synchronizability, instead of language synchronizability. It is slightly more involved due to the possible existence of stable traces of $\system_{\A}$
that are not ``catched'' by a stable synchronous trace of $\system{\A}$.
This is actually only a minor problem, and we will actually fix it with the following extra hypothesis on the FIFO automaton $\A$.

\begin{defi}
A FIFO automaton $\A$ is \emph{good for reduction} if the only stable trace of $\A$
is the empty trace.
\end{defi}

Note that the FIFO automaton $\A$ that we defined in the proof of the undecidability of the message reception problem is good for reduction: indeed, after the
first row of the tiling has been queued, the automaton always queues a new tile
right after it has dequeued a tile, or queues the marker of the end of the row
(\$)
right after it dequeues it. So the buffer always contains either at least
one tile or the \$ marker, except in the initial configuration.

\begin{lem}
\label{lem:reduction-message-reception-synchronizability}
Assume that $\A$ is good for reduction.
Then $\system_{\A,\specialmsg}''$ is not synchronizable
iff there is a trace $\trace\in\traces{}{\A}$ such that
$?\specialmsg$ occurs in $\trace$.
\end{lem}

\begin{proof}
Let us show that, under the hypothesis that $\A$ is good for reduction,
$\system_{\A,\specialmsg}''$ is synchronizable if and only if
$\system_{\A,\specialmsg}'$ is language synchronizable, which, by Lemma~\ref{lem:reduction-message-reception-language-synchronizability} will entail
what we need to prove.

($\Rightarrow$) Let us assume that
$\system_{\A,\specialmsg}''$ is synchronizable and let us show that
$\system_{\A,\specialmsg}''$ is language synchronizable.
This implication actually holds for any system $\system$.
Indeed,
if $\system$ is synchronizable, then
$\sendtraces{}{\system}\setminus\sendtraces{0}{\system}=\emptyset$.
Since $\sendtracesweak{}{\system}\subseteq \sendtraces{}{\system}$,
we have in particular
$$
\sendtracesweak{}{\system}\setminus\sendtraces{0}{\system}=\emptyset.
$$
By definition, $\sendtraces{0}{\system}=\sendtracesweak{0}{\system}\cup\mathsf{Stab}$, where $\mathsf{Stab}$ is a set of pairs $(\sendprojection{\trace},\configuration)$; such pairs do not belong to $\sendtracesweak{}{\system}$,
so
$$
\sendtracesweak{}{\system}\setminus\sendtraces{0}{\system}=
\sendtracesweak{}{\system}\setminus\sendtracesweak{0}{\system}.
$$
As a consequence, $\sendtracesweak{}{\system}\setminus\sendtracesweak{0}{\system}=\emptyset$, and since $\sendtracesweak{0}{\system}\subseteq \sendtracesweak{}{\system}$, we finally get
$$
\sendtracesweak{}{\system} = \sendtraces{0}{\system}.
$$

($\Leftarrow$) Let us assume that
$\sendtracesweak{k}{\system_{\A,\specialmsg}''}=\sendtracesweak{0}{\system_{\A,\specialmsg}''}$, and let us show that
$\sendtraces{k}{\system_{\A,\specialmsg}''}=\sendtraces{0}{\system_{\A,\specialmsg}''}$. The inclusion $\sendtraces{0}{\system_{\A,\specialmsg}''}\subseteq \sendtraces{}{\system_{\A,\specialmsg}''}$ holds for any system. Let us therefore
show that
$\sendtraces{k}{\system_{\A,\specialmsg}''}\subseteq \sendtraces{0}{\system_{\A,\specialmsg}''}$. Since
$\sendtraces{k}{\system_{\A,\specialmsg}''}
=
\sendtraces{k}{\system_{\A}}
\cup
\sendtraces{k}{\system_{\A,\specialmsg}'}$
for all $k\geq 0$,
we have to show that
$\sendtraces{}{\system_{\A}}\subseteq\sendtraces{0}{\system_{\A}}\cup\sendtraces{0}{\system_{\A,\specialmsg}'}$
and
$\sendtraces{}{\system_{\A,\specialmsg}'}\subseteq\sendtraces{0}{\system_{\A}}\cup\sendtraces{0}{\system_{\A,\specialmsg}'}$.
\begin{itemize}
\item $\sendtraces{}{\system_{\A,\specialmsg}'}\subseteq\sendtraces{0}{\system_{\A}}\cup\sendtraces{0}{\system_{\A,\specialmsg}'}$ follows from Lemma~\ref{lem:system-2-synchronizable}.
\item Let us show that $\sendtraces{}{\system_{\A}}\subseteq\sendtraces{0}{\system_{\A}}\cup\sendtraces{0}{\system_{\A,\specialmsg}'}$.
Since $\system_{\A,\specialmsg}''$ is language synchronizable by hypothesis,
we have in particular that $\sendtracesweak{}{\A}\subseteq \sendtracesweak{0}{\system_{\A}}\cup\sendtracesweak{0}{\system_{\A,\specialmsg}'}$. So we only need
to prove that for all stable trace $\trace$ of $\system_{\A}$, there is
a stable synchronous trace $\trace'$ of
$\system_{\A,\specialmsg}''$ leading to the same configuration and such that
$\sendprojection{\trace}=\sendprojection{\trace'}$.
We will actually show that the only stable traces of $\system_{\A}$
are synchronous, and therefore we can even take $\trace'=\trace$.

Let $\trace\in\traces{}{\system_{\A}}$ a stable trace be fixed,
and let us show that $\trace$ is synchronous.
By Lemma~\ref{lem:fifo-auto-encoding}, there is a trace
$\trace_0\in\traces{}{\A}$ and a message $a\in\Sigma$
such that either $\trace=h(\trace_0)$, or
$\trace=h(\trace_0)\cdot{}!?\tagmsg{a}{1}{3}$, or
$\trace=h(\trace_0)\cdot{}!?\tagmsg{a}{1}{3}\cdot{}!?\tagmsg{a}{3}{2}$.
By definition of $h$, if $\trace$ is stable, then $\trace_0$ is stable too.
Since $\A$ is good for reduction, $\trace_0$ must be the empty trace.
So either $\trace$ is the empty trace, or $\trace=!?\tagmsg{a}{1}{3}$,
or $\trace=!?\tagmsg{a}{1}{3}\cdot{}!?\tagmsg{a}{3}{2}$. In all cases, $\trace$
is a synchronous trace, which ends the proof.
\end{itemize}
\end{proof}

\begin{thm}\label{thm:undecidability}
Synchronizability (resp. language synchronizability) is undecidable.
\end{thm}

\begin{proof}
Let a FIFO automaton $\A$ that is good for reduction over the message alphabet
$\Sigma$
and let  a message $\specialmsg\in\Sigma$ be fixed.
By Lemma~\ref{lem:reduction-message-reception-language-synchronizability},
$\system_{\A,\specialmsg}''$ is
(language) synchronizable
iff $\sendtracesweak{k}{\system_{\A}}\setminus\sendtracesweak{0}{\system_{\A,\specialmsg}'}=\emptyset$. By Lemma~\ref{lem:reduction-message-reception-synchronizability},
$\sendtracesweak{k}{\system_{\A}}\setminus\sendtracesweak{0}{\system_{\A,\specialmsg}'}=\emptyset$ iff
there is no trace $\trace$ such that
$\trace\cdot{}?\specialmsg\in\traces{}{\A}$. By
Lemma~\ref{lem:undecidability-FIFO-automaton}, this is an undecidable problem.
\end{proof}
%
\section{The case of oriented rings\label{sec:bipartite}}
In the previous section we established the undecidability
of synchronizability for systems with (at least) three peers. In this section,
we show that this result is tight, in the sense that synchronizability
is decidable if $G_M$ is an oriented ring, in particular if the system involves
two peers only. This relies on the fact
that 1-synchronizability implies synchronizability for such systems.
\ificalp
This property is highly non-trivial, and below we only sketch the main steps of the proof, identifying
when the hypothesis on the ring topology becomes necessary.
\else
In order to show this result, we first establish some
confluence properties on traces for arbitrary topologies.
With the help of this confluence properties, we can state a trace normalization
property that is similar to the one that was used in~\cite{BasuBO12} and
for half-duplex systems~\cite{CeceF05}. This trace normalization property implies
that $1$-synchronizable systems on oriented rings have no unspecified
receptions\footnote{An unspecified reception occurs when a process $P$
is in a receiving state, some messages awaits for $P$ receiving them,
but they are not the ones that $P$ may dequeue. See~\cite{CeceF05} for formal definitions.},
and their reachability set is channel-recognizable.
Finally, this trace normalization property leads to a proof that $1$-synchronizability
implies synchronizability when $G_M$ is an oriented ring.
\fi

\ificalp
The starting point is a confluence property on arbitrary $1$-synchronizable systems.
\medskip
\else
\subsection{Confluence properties}
\fi

The following confluence property holds for any synchronizable system
(see also Fig~\ref{fig:send-send-diagram}).

\begin{lem}[Weak commutativity]\label{lem:send-send-diagram}
Let $\system$ be a $1$-synchronizable system.
Let $\trace\in\traces{0}{\system}$
and $a,b\in\messages$ be such that

\begin{enumerate}
\item
$\trace\cdot{} !a\in\traces{1}{\system}$,
\item
$\trace\cdot{} !b\in\traces{1}{\system}$, and
\item
$\msgsrc(a)\neq\msgsrc(b)$.
\end{enumerate}
If $\tracebis_1,\tracebis_2$ are any two of the six different shuffles
of $!a\cdot{}?a$ with $!b\cdot{}?b$, then
$\trace\cdot{}\tracebis_1\in\traces{}{\system}$,
$\trace\cdot{}\tracebis_2\in\traces{}{\system}$ and
$\trace\cdot{}\tracebis_1\systemequiv\trace\cdot{}\tracebis_2$.
\end{lem}

\begin{figure}
\begin{tikzpicture}[scale=.7]

\begin{scope}[every node/.style={scale=.2}]
\node (0) at (0,0) {};
\node (1) at (-1,-1) {};
\node (2) at (1,-1) {};
\node (3) at (-2,-2) {};
\node (4) at (0,-2) {};
\node (5) at (2,-2) {};
\node (6) at (-1,-3) {};
\node (7) at (1,-3) {};
\node (8) at (0,-4) {};
\end{scope}

\node at (0,-1) {$\causalequiv$};

\draw[->] (0) edge node [above left] {$!a$} (1);
\draw[->] (0) edge node [above right] {$!b$} (2);
\draw[->,dashed] (1) edge node [pos=.3,below] {$!b$} (4);
\draw[->,dashed] (1) edge node [above left] {$?a$} (3);
\draw[->,dashed] (2) edge node [pos=.3,below] {$!a$} (4);
\draw[->,dashed] (2) edge node [above] {$?b$} (5);
\draw[->,dashed] (3) edge node [below left] {$!b$} (6);
\draw[->,dashed] (4) edge node [pos=.3,below] {$?a$} (6);
\draw[->,dashed] (4) edge node [pos=.3,below] {$?b$} (7);
\draw[->,dashed] (5) edge node [below right] {$!a$} (7);
\draw[->,dashed] (6) edge node [below left] {$?b$} (8);
\draw[->,dashed] (7) edge node [below right] {$?a$} (8);

\end{tikzpicture}
\caption{\label{fig:send-send-diagram}Diagrammatic representation of Lemma~\ref{lem:send-send-diagram}}
\end{figure}

\begin{rem}
This lemma should not be misunderstood as a consequence of
causal equivalence. Observe indeed that
the square on top of the diagram is the only square that commutes for causal equivalence. The three other squares only
commute with respect to $\systemequiv$, and they commute for
$\causalequiv$ only if some extra assumptions on $a$ and $b$ are made. For instance, the left square does commute for
$\causalequiv$ if and only if $\msgdst(a)\neq \msgsrc(b)$.
\end{rem}

\ificalp\else
Before we prove Lemma~\ref{lem:send-send-diagram}, let us first prove the following.

\begin{lem}
\label{lem:send-send-diagram-aux}
Let $a,b$ be two messages such that $\msgsrc(a)\neq\msgsrc(b)$.
Then for all peers $i$, for all shuffle $\tracebis$ 
of $!a\cdot{}?a$ with $!b\cdot{} ?b$,
either $\machineprojection{i}{\tracebis}=\machineprojection{i}{!?a\cdot{}!?b}$
or $\machineprojection{i}{\tracebis}=\machineprojection{i}{!?b\cdot{}!?a}$.
\end{lem}

\begin{proof}
Let $i$ and $\tracebis$ be fixed. Since $\msgsrc(a)\neq\msgsrc(b)$,
it is not the case that both $!a$ and $!b$ occur in 
$\machineprojection{i}{\tracebis}$.
By symmetry, let us assume that $!b$ does not occur in $\machineprojection{i}{\tracebis}$. We consider two cases:
\begin{enumerate}
\item
Let us assume that $?a$ does not occur in 
$\machineprojection{i}{\tracebis}$. Then 
$\machineprojection{i}{\tracebis}\in\{!a,?b,!a\cdot{}?b\}$, and in
all cases, 
$\machineprojection{i}{\tracebis}=\machineprojection{i}{!?a\cdot{}!?b}$.
\item 
Let us assume that $?a$ occurs in 
$\machineprojection{i}{\tracebis}$. Then $!a$ does not
occur in $\machineprojection{i}{\tracebis}$, 
therefore $\machineprojection{i}{\tracebis}$ contains only receptions, and
$\machineprojection{i}{\tracebis}\in\{?a,?a\cdot{}?b,?b,?b\cdot{}?a\}$.
In every case, either 
$\machineprojection{i}{\tracebis}=\machineprojection{i}{!?a\cdot{}!?b}$
or
$\machineprojection{i}{\tracebis}=\machineprojection{i}{!?b\cdot{}!?a}$.
\qedhere
\end{enumerate}
\end{proof}

Let us prove now Lemma~\ref{lem:send-send-diagram}.

\begin{proof}
Observe first that, since $\msgsrc(a)\neq\msgsrc(b)$, $\trace\cdot{}!a\cdot{}!b\in\traces{1}{\system}$
and $\trace\cdot{}!b\cdot{}!a\in\traces{1}{\system}$, and
since $\system$ is $1$-synchronizable, 
$\trace\cdot{}!?a\cdot!?b\in\traces{0}{\system}$ 
and 
$\trace\cdot{}!?b\cdot!?a\in\traces{0}{\system}$. 
By Lemma~\ref{lem:send-send-diagram-aux}, it follows that
for all shuffle $\tracebis$ of $!a\cdot{}?a$ with $!b\cdot{}?b$, 
$\trace\cdot{}\tracebis\in\traces{1}{\system}$. 
It remains to show that
$$
\mbox{for all two shuffles $\tracebis$, $\tracebis'$ of
$!a\cdot{}?a$ with $!b\cdot{}?b$,}\quad 
\trace\cdot{}\tracebis\systemequiv\trace\cdot{}\tracebis'.
\qquad (P) 
$$

Let $\trace_{ab}=!a\cdot{}!b\cdot{}?a\cdot{}?b$ and
$\trace_{ba}=!b\cdot{}!a\cdot{}?a\cdot{}?b$, and let $\tracebis$ be a
shuffle of $!a\cdot{}?a$ with $!b\cdot{}?b$.
Since $\system$ is $1$-synchronizable, the stable configuration that $\trace\cdot{}\tracebis$ leads 
to only depends on
the order in which the send actions $!a$ and $!b$ 
are executed in $\tracebis$, i.e. either 
$\trace\cdot{}\tracebis\systemequiv\trace_{ab}$ or 
$\trace\cdot{}\tracebis\systemequiv\trace_{ba}$.
Moreover, $\trace_{ab}\causalequiv\trace_{ba}$, hence $(P)$.
\end{proof}
\fi

\ificalp
\noindent Lemma~\ref{lem:send-send-diagram} then generalizes to arbitrary sequences of send
actions with rather technical arguments.
\else
Our aim now is to generalize Lemma~\ref{lem:send-send-diagram} to arbitrary sequences of send
actions (see Lemma~\ref{lem:send-send-diagram-star} below and the corresponding diagram).
For this, we need to reason by induction on the length of the sequence of send actions.
The first step is to establish the following property: a synchronous trace followed by a sequence of send actions can be completed to form a fully synchronous trace.
 \begin{lem}\label{lem:lemaux-0}
Let $\system$ be a $1$-synchronizable system.
Let $\trace\in\traces{0}{\system}$ and $a_1,\cdots{},a_n\in\messages$ be such that
\begin{enumerate}
\item $\trace\cdot{}!a_1\cdots{}!a_n\in\traces{n}{\system}$
\item $\msgsrc(a_i)=\msgsrc(a_j)$ for all $i,j\in\{1,\dots,n\}$.
\end{enumerate}
Then $\trace\cdot{}!?a_1\cdots{}!?a_n\in\traces{0}{\system}$.
\end{lem}
\begin{proof}
By induction on $n$. Let $a_1,\dots,a_{n+1}$ be fixed,
and let $\trace_n=\trace\cdot{}!?a_1\cdots{}!?a_n$.
By induction hypothesis, $\trace_n\in\traces{0}{\system}$.
Let $\trace_{n+1}'=\trace_n\cdot{}!a_{n+1}$.
Then
\begin{itemize}
\item $\machineprojection{i}{\trace_{n+1}'}=\machineprojection{i}{\trace_n}$
for all $i\neq\msgsrc(a_{n+1})$, and $\trace_n\in\traces{}{\system}$
\item for $i=\msgsrc(a_{n+1})$,
$\machineprojection{i}{\trace_{n+1}'}=\machineprojection{i}{\trace\cdot{}!a_1\cdots{}!a_{n+1}}$ and $\trace\cdot{}!a_1\cdots{}!a_n\in\traces{}{\system}$
\item $\trace_{n+1}'$ is $1$-bounded FIFO
\end{itemize}
therefore $\trace_{n+1}'\in\traces{1}{\system}$.

By $1$-synchronizability, it follows that $\trace_{n+1}'\cdot{}?a_{n+1}\in\traces{0}{\system}$.
\end{proof}

The second step is a confluence property that allows to commute a synchronization on one message and a sequence of synchronizations on other messages with different senders (see also the diagram on Figure~\ref{fig:lem-lemaux2}).

\begin{figure}
\begin{tikzpicture}[scale=1.2]
\node (0) at (0,0) {};
\node (1) at (-1,-1) {};
\node (2) at (1,-1) {};
\node (3) at (0,-2) {};
\draw[->] (0) edge node [above left] {$!?a$} (1);
\draw[->] (0) edge node [above right] {$!?b_1\cdots !?b_n$} (2);
\draw[->,dashed] (1) edge node [below left] {$!?b_1\cdots !?b_n$} (3);
\draw[->,dashed] (2) edge node [below right] {$!?a$} (3);
\end{tikzpicture}

\caption{\label{fig:lem-lemaux2}Diagram of Lemma~\ref{lem:lemaux-2}}
\end{figure}
\begin{lem}\label{lem:lemaux-2}
Let $\system$ be a $1$-synchronizable system.
Let $\trace\in\traces{0}{\system}$ and $a,b_1,\dots,b_n\in\messages$ be such that
\begin{enumerate}
\item $\trace\cdot{}!?a\in\traces{0}{\system}$
\item $\trace\cdot{}!?b_1\cdots{}!?b_n\in\traces{0}{\system}$
\item $\msgsrc(a)\neq\msgsrc(b_i)$ for all $i\in\{1,\dots,n\}$.
\end{enumerate}
Then the following holds
\begin{itemize}
\item
$\trace\cdot{}!?a\cdot{}!?b_1\cdots{}!?b_n\in\traces{0}{\system}$,
\item
$\trace\cdot{}!?b_1\cdots{}!?b_n\cdot{}!?a\in\traces{0}{\system}$,
and
\item
$\trace\cdot{}!?a\cdot{}!?b_1\cdots{}!?b_n\systemequiv
\trace\cdot{}!?b_1\cdots{}!?b_n\cdot{}!?a$.
\end{itemize}
\end{lem}

\begin{proof}
By induction on $n$.
Let $a,b_1\dots,b_{n+1}$ be fixed,
let $\trace_n=\trace\cdot{}!?b_1\cdots{}!?b_n$.
By induction hypothesis, $\trace_n\cdot{}!?a\in\traces{0}{\system}$, and
by hypothesis $\trace_n\cdot{}!?b_{n+1}\in\traces{0}{\system}$.
By Lemma~\ref{lem:send-send-diagram}, $\trace_n\cdot{}!?a\cdot{}!?b_{n+1}\in\traces{0}{\system}$,
$\trace_n\cdot{}!?b_{n+1}\cdot{}!?a\in\traces{0}{\system}$, and
$$
\trace_n\cdot{}!?a\cdot{}!?b_{n+1}\systemequiv \trace_n\cdot{}!?b_{n+1}\cdot{}!?a.
$$
On the other hand,
by induction hypothesis, $\trace_n\cdot{}!?a\systemequiv\trace\cdot{}!?a\cdot{}!?b_1\cdots{}!?b_n$, and by right congruence of $\systemequiv$
$$
\trace_n\cdot{}!?a\cdot{}!?b_{n+1}\systemequiv\trace\cdot{}!?a\cdot{}!?b_1\cdots{}!?b_{n+1}
$$
By transitivity of $\systemequiv$, we can relate the two right members of the above
identities, \emph{i.e.}
$$
\trace_n\cdot{}!?b_{n+1}\cdot{}!?a\systemequiv\trace\cdot{}!?a\cdot{}!?b_1\cdots{}!?b_{n+1}
$$
which shows the claim.
\end{proof}

The next lemma expresses the following, rather technical property: considering two sequences of
synchronizations that are orthogonal (with different senders), it is possible to combine them in a single synchronous trace by shuffling the synchronization in any order. Diagramatically, it expresses that all paths that result from a shuffle inside the diamond lead to the same configuration (see Figure~\ref{fig:lemaux3}).

\begin{figure}
\begin{tikzpicture}
\node (0) at (0,0) {};
\node (1) at (-1,-1) {};
\node (2) at (1,-1) {};
\node (3) at (0,-2) {};
\draw[->] (0) edge node [above left] {$!?a_1\cdots!?a_n$} (1);
\draw[->] (0) edge node [above right] {$!?b_1\cdots !?b_m$} (2);
\draw[->,dashed] (1) edge node [below left] {$!?b_1\cdots !?b_m$} (3);
\draw[->,dashed] (2) edge node [below right] {$!?a_1\cdots!?a_n$} (3);
\draw[->,dashed] (0) edge node [right] {\tiny$!?c_1\cdots !?c_{n+m}$} (3);
\end{tikzpicture}
\caption{\label{fig:lemaux3}Diagram of Lemma~\ref{lem:lemaux-3}}
\end{figure}

\begin{lem}\label{lem:lemaux-3}
Let $\system$ be a $1$-synchronizable system.
Let $\trace\in\traces{0}{\system}$ and $a_1,\dots,a_n,b_1,\dots,b_m\in\messages$ be
 such that
\begin{enumerate}
\item $\trace\cdot{}!?a_1\cdots{}!?a_n\in\traces{0}{\system}$
\item $\trace\cdot{}!?b_1\cdots{}!?b_m\in\traces{0}{\system}$
\item $\msgsrc(a_i)\neq\msgsrc(b_j)$ for all $i\in\{1,\dots,n\}$, $j\in\{1,\dots,m\}$
\end{enumerate}
Then for all shuffle $c_1\dots c_{m+n}$ of $a_1\cdots{} a_n$ with $b_1\cdots{}b_m$,
\begin{itemize}
\item
$\trace\cdot{}!?c_1\cdots{}!?c_{n+m}\in\traces{0}{\system}$, and
\item
$\trace\cdot{}!?a_1\cdots{}!?a_n\cdot{}!?b_1\cdots{}!?b_m\systemequiv
\trace\cdot{}!?c_1\cdots{}!?c_m$.
\end{itemize}
\end{lem}

\begin{proof}
By induction on $n+m$.
Let $a_1,\dots,a_{n},b_1\dots,b_{m}$ be fixed, and let $c_1\cdots{}c_{n+m}$
be a shuffle of $a_1\cdots{}a_{n}$ with $b_1\cdots b_m$.
\begin{itemize}
\item
Assume that $c_1=a_1$. Let $\trace'=\trace\cdot{}!?a_1$.
By Lemma~\ref{lem:lemaux-2}, $\trace'\cdot{}!?b_1\cdots{}!?b_m\in\traces{0}{\system}$,
and by hypothesis $\trace'\cdot{}!?a_2\cdots{}!?a_n\in\traces{0}{\system}$,
so we can use the induction hypothesis with $(a'_1,\dots,a'_{n-1})=(a_2,\dots,a_{n})$.
We get $\trace'\cdot{}!?c_2\cdots{}!?c_n\in\traces{0}{\system}$, and
$$
\trace'\cdot{}!?c_2\cdots{}!?c_n\systemequiv\trace'\cdot{}!?a_2\cdots{}!?a_n\cdot{}!?b_1\cdots{} !?b_m
$$
which shows the claim.
\item Assume that $c_1=b_1$. Then by the same arguments,
$$
\trace\cdot{}!?c_1\cdots{}!?c_n\systemequiv\trace\cdot{}!?b_1\cdots{}!?b_m\cdot{}!?a_1\cdots{} !?a_n
$$
Since this holds for all shuffle $c_1,\dots,c_{n+m}$, this also holds for
$c_1=a_1,\dots,c_n=a_n,c_{n+1}=b_1,\cdots,c_{n+m}=b_m$, which shows the claim.
\qedhere
\end{itemize}
\end{proof}

The next lemma generalizes Lemma~\ref{lem:lemaux-0}: a sequence of send following a synchronous trace can be completed in a synchronous trace, regardless whether these sends are from the same sender or not.

\begin{lem}\label{lem:lemaux-4}
Let $\system$ be a $1$-synchronizable system.
Let $\trace\in\traces{0}{\system}$ and $m_1,\cdots{},m_n\in\messages$ be such that
$\trace\cdot{}!m_1\cdots{}!m_n\in\traces{n}{\system}$
Then $\trace\cdot{}!?m_1\cdots{}!?m_n\in\traces{0}{\system}$.
\end{lem}
\begin{proof}
By induction on $n$. Let $m_1,\dots,m_{n}$ be fixed with $n\geq 1$.
There are two subsequences $a_1,\dots,a_r$ and $b_1,\dots,b_m$ such that
\begin{itemize}
\item $\msgsrc(a_\ell)=\msgsrc(m_1)$ for all $\ell\in\{1,\dots,r\}$,
\item $\msgsrc(b_\ell)\neq\msgsrc(m_1)$ for all $\ell\in\{1,\dots,m\}$,
\item $m_1\cdots m_n$ is a shuffle of $a_1\cdots{}a_r$ with $b_1\cdots{} b_m$
\end{itemize}
By hypothesis, $\trace\cdot{}!a_1\cdots{}!a_r\in\traces{}{\system}$
and $\trace\cdot{}!b_1\cdots{}!b_m\in\traces{}{\system}$.
By Lemma~\ref{lem:lemaux-0},
 $\trace\cdot{}!?a_1\cdots{}!?a_r\in\traces{0}{\system}$,
and by induction hypothesis $\trace\cdot{}!?b_1\cdots{}!?b_m\in\traces{0}{\system}$,
and finally by Lemma~\ref{lem:lemaux-3}
$\trace\cdot{}!?m_1\cdots{}!?m_n\in\traces{0}{\system}$.
\end{proof}

The next lemma, the last one before the main lemma we aim at, is a purely combinatorial property
that does not have anything to do with synchronizability. It says that if a trace is a shuffle of two synchronous traces, and if it projected on a given machine, then this projection looks like the projection of a synchronous trace that is a shuffle of the original messages.
\begin{lem}\label{lem:lemaux-5}
Let $a_1,\dots,a_n,b_1,\cdots,b_m\in\messages$, and let
$\trace$ be a shuffle of $!?a_1\cdots!?a_n$ with $!?b_1\cdots{}!?b_m$.
Then for all $i\in\{1,\dots,\nbpeers\}$ there is
a shuffle $c_1\cdots c_{n+m}$ of $a_1\cdots a_n$ with
$b_1\cdots b_m$ such that $\machineprojection{i}{\trace}=\machineprojection{i}{!?c_1\cdots{}!?c_{n+m}}$.
\end{lem}

\begin{proof}
  Let us fix $\trace=\action_1\cdots\action_{2(n+m)}$ a shuffle $!?a_1\cdots!?a_n$ with $!?b_1\cdots{}!?b_m$.
  Consider the trace $\trace'=\trace_1'\cdots\trace_{2(n+m)}'$
  where, for all $k\in\{1,\dots, 2(n+m)\}$, $\trace_k'$  is defined as follows:
  \begin{itemize}
  \item if there are $m, j$ such that
    $\action_k= !\tagmsg{m}{i}{j}$ or $\action_k= ?\tagmsg{m}{j}{i}$,
    then $\trace_k'=!?m$
  \item otherwise, if $\action_k=!m_k$, then $\trace_i=!?m_k$, else
    $\action_i=?m_k$ and $\trace_k'=\epsilon$
  \end{itemize}
  Then by construction $\trace'$ is of the form $!?c_1\cdots !?c_{n+m}$
  with $c_1,\dots,c_{n+m}$ a shuffle of $a_1,\dots,a_n$ with $b_1,\dots,b_m$.
  Moreover, $\onmachine{i}{\trace'}=\onmachine{i}{\trace}$, which ends the proof.
\qedhere

\end{proof}
We are now ready to generalise Lemma~\ref{lem:send-send-diagram} to an arbitrary long sequence of send actions, which is the main property we wanted to establish with this long serie of lemmas.
\fi

\begin{figure}
\begin{tikzpicture}[scale=.7]

\begin{scope}[every node/.style={scale=.2}]
\node (0) at (0,0) {};
\node (1) at (-1,-1) {};
\node (2) at (1,-1) {};
\node (3) at (-2,-2) {};
\node (4) at (0,-2) {};
\node (5) at (2,-2) {};
\node (6) at (-1,-3) {};
\node (7) at (1,-3) {};
\node (8) at (0,-4) {};
\end{scope}

\node at (0,-1) {$\causalequiv$};

\draw[->] (0) edge node [above left] {$!a_1\cdots !a_n$} (1);
\draw[->] (0) edge node [above right] {$!b_1\cdots !b_m$} (2);
\draw[->,dashed] (1) edge node [pos=.3,below] {} (4);
\draw[->,dashed] (1) edge node [above left] {} (3);
\draw[->,dashed] (2) edge node [pos=.3,below] {} (4);
\draw[->,dashed] (2) edge node [above] {} (5);
\draw[->,dashed] (3) edge node [below left] {} (6);
\draw[->,dashed] (4) edge node [pos=.3,below] {} (6);
\draw[->,dashed] (4) edge node [pos=.3,below] {} (7);
\draw[->,dashed] (5) edge node [below right] {} (7);
\draw[->,dashed] (6) edge node [below left] {} (8);
\draw[->,dashed] (7) edge node [below right] {} (8);

\end{tikzpicture}
\caption{\label{fig:send-send-diagram-star}Diagrammatic representation of Lemma~\ref{lem:send-send-diagram-star}}
\end{figure}

\begin{lem}[Strong commutativity]\label{lem:send-send-diagram-star}
Let $\system$ be a $1$-synchronizable system.
Let
$a_1,\dots,a_n,b_1,\dots b_m\in\messages$
and
$\trace\in\traces{0}{\system}$
be such that

\begin{enumerate}
\item
$\trace\cdot{} !a_1\cdots{}!a_n\in\traces{n}{\system}$,
\item
$\trace\cdot{} !b_1\cdots{} !b_m\in\traces{m}{\system}$, and
\item
$\msgsrc(a_i)\neq\msgsrc(b_j)$
for all $i\in\{1,\dots,n\}$ and $j\in\{1,\dots,m\}$.
\end{enumerate}
Then for any two different shuffles $\tracebis_1,\tracebis_2$ of
$!?a_1\cdot{}!?a_2\cdots{}!?a_n$ with $!?b_1\cdot{}!?b_2\cdots{}!?b_m$,
it holds that
$\trace\cdot{}\tracebis_1 \in\traces{}{\system}$ , $\trace\cdot{}\tracebis_2\in\traces{}{\system}$ and
$\trace\cdot{}\tracebis_1\systemequiv\trace\cdot{}\tracebis_2$.
\end{lem}

\ificalp\else
\begin{proof}
Let $\trace\in\traces{0}{\system}$ and $a_1,\dots,a_{n},b_1,\dots,b_m$,
be fixed.
Let $\tracebis$ be a shuffle of
$!?a_1\cdots{}!?a_n$ with $!?b_1\cdots{}!?b_m$.
We want to show that $\trace\cdot{}\tracebis\in\traces{}{\system}$.
Clearly, $\trace\cdot{}\tracebis\in\traces{}{\system}$ is a FIFO trace.
Therefore, it is enough to find for all $i\in\{1,\dots,\nbpeers\}$
a trace $\trace_i$ such that
\begin{equation}\label{eq:goal}
\trace_i\in\traces{}{\system}\qquad \mbox{and} \qquad
\machineprojection{i}{\trace\cdot{}\tracebis}=\machineprojection{i}{\trace_i}.
\end{equation}
Let $i\in\{1,\dots,\nbpeers\}$ be fixed, and let us construct $\trace_i$ that validates
$(\ref{eq:goal})$.
By hypothesis
$$
\trace\cdot{}!a_1\cdots{}!a_n\in\traces{}{\system}
\mbox{ and }
\trace\cdot{}!b_1\cdots{}!b_n\in\traces{}{\system}
$$
therefore, by Lemma~\ref{lem:lemaux-4},
\begin{equation}\label{eq:dribble}
\trace\cdot{}!?a_1\cdots{}!?a_n\in\traces{0}{\system}
\mbox{ and }
\trace\cdot{}!?b_1\cdots{}!?b_n\in\traces{0}{\system}.
\end{equation}
On the other hand, by Lemma~\ref{lem:lemaux-5},
there is a shuffle $c_1\dots c_{n+m}$
of $a_1\cdots a_n$ with $b_1\cdots b_m$ such that
\begin{equation}\label{eq:shoot}
\machineprojection{i}{\tracebis}=\machineprojection{i}{!?c_1\cdots{}!?c_{n+m}}
\end{equation}
Let $\trace_i=\trace\cdot{}!?c_1\cdots{}!?c_{n+m}$.
By Lemma~\ref{lem:lemaux-3} and $(\ref{eq:dribble})$,
$\trace_i\in\traces{0}{\system}$, and by
$(\ref{eq:shoot})$, the second part of $(\ref{eq:goal})$ holds.
\end{proof}
\fi

\ificalp
For the rest of the proof, the hypothesis on the communication topology being an oriented ring becomes necessary.
We follow the rough idea in~\cite{BasuBO12}, also used for half-duplex
systems~\cite{CeceF05}, and show a trace normalization property.
\else
\subsection{Trace normalization}

In this section and the next one, it will be necessary to assume that the communication topology
is an oriented ring.
\fi

\begin{defi}[Normalized trace]
A $M$-trace $\trace$ is \emph{normalized}
if there is a synchronous $M$-trace
$\trace_0$, $n\geq 0$, and messages $a_1,\dots,a_n$ such that
$\trace=\trace_0\cdot{}!a_1\cdots{}!a_n$.
\end{defi}

\noindent\begin{minipage}{.8\textwidth} \setlength{\parfillskip}{0pt}
\begin{lem}[Trace Normalization]\label{lem:trace-normalization}
Assume $M$ is such that the communication topology $G_M$ is an oriented ring.
Let $\system=\systemlist{\peer_1,\dots,\peer_{\nbpeers}}$ be a $1$-synchronizable $M$-system.
For all $\trace\in\traces{}{\system}$,
there is a normalized trace $\normaltrace(\trace)
\in\traces{}{\system}$ such that
$\trace\systemequiv \normaltrace(\trace)$.
\end{lem}
\end{minipage}
\raisebox{1mm}{\begin{minipage}{0.1\textwidth}
\begin{tikzpicture}[scale=1.2]

\node (0) at (0,0) {};
\node (1) at (1,-1) {};
\node (2) at (0,-1) {};

\draw[->] (0) edge node [above right] {$\trace$} (1);
\draw[->,dashed] (2) edge node [below] {$!m_1\cdots{}!m_k$} (1);
\draw[->,dashed] (0) edge node [left] {$\trace_0$} (2);

\end{tikzpicture}

\end{minipage}} \\[1mm]

\begin{proof}
By induction on $\trace$. Let $\trace=\trace'\cdot{}\action$,
be fixed. Let us assume by induction hypothesis that there is
a normalized trace $\normaltrace(\trace')\in\traces{}{\system}$
such that $\trace'\systemequiv\normaltrace(\trace')$.
Let us reason
by case analysis on the last action $\action$ of $\trace$.
The easy case is when $\action$ is a send
action: then,
$\normaltrace(\trace')\cdot{}\action$ is a normalized trace,
and $\normaltrace(\trace')\cdot{}\action\systemequiv\trace'\cdot{}\action$
by right congruence of $\systemequiv$.
The difficult case is when $\action$ is $?a$ for some $a\in\messages$.
Let $i=\msgsrc(a)$, $j=\msgdst(a)$, \emph{i.e.} $i+1=j$ mod $\nbpeers$.
By the definitions of a normal trace and $\causalequiv$, there are
$\trace_0'\in\traces{0}{\system}$,
$a_1,\dots,a_n,b_1,\dots,b_m\in\messages$ such that
$$
\normaltrace(\trace')\causalequiv \trace_0'\cdot{}!a_1\cdots{}!a_n\cdot{}!b_1\cdots{}!b_m
$$
with $\msgsrc(a_k)=i$ for all $k\in\{1,\dots,n\}$ and $\msgsrc(b_k)\neq i$
for all $k\in\{1,\dots,m\}$.
Since $G_M$ is an oriented ring, $\msgdst(a_1)=j$, therefore $a_1=a$
(because by hypothesis $j$ may receive $a$ in the configuration that $\normaltrace(\trace')$ leads to).
Let
$
\normaltrace(\trace)=\trace_0'\cdot{}!a\cdot{}?a\cdot{}!b_1\cdots !b_m
\cdot{} !a_2\cdots{} !a_n
$
and let us show that $\normaltrace(\trace)\in\traces{}{\system}$
and $\trace\systemequiv\normaltrace(\trace)$.

Since $\normaltrace(\trace')\in\traces{}{\system}$, we have
in particular that $\trace_0'\cdot{} !a\in\traces{1}{\system}$ and
$\trace_0'\cdot{}!b_1\cdots!b_n\in\traces{}{\system}$. Consider the two traces
$$
\begin{array}{ll}
\tracebis_1= & \trace_0'\cdot{}!a\cdot{}?a\cdot{}!b_1\cdots !b_n\cdot{} ?b_1\cdots ?b_n\\
\tracebis_2= & \trace_0'\cdot{}!a\cdot{}!b_1\cdots !b_n\cdot{} ?a\cdot{}?b_1\cdots ?b_n.
\end{array}
$$
By Lemma~\ref{lem:send-send-diagram-star},
$\tracebis_1,\tracebis_2\in\traces{}{\system}$ and both lead to the same configuration,
and in particular to the same control state $q$ for peer $j$.
The actions $?b_1$, $?b_2,\ldots ?b_n$
are not executed by peer $j$ (because $\msgsrc(m)\neq i$ implies $\msgdst(m)\neq j$ on an
oriented ring),
so the two traces
$$
\begin{array}{ll}
\tracebis_1'= & \trace_0'\cdot{}!a\cdot{}?a\cdot{}!b_1\cdots !b_n\\
\tracebis_2'= & \trace_0'\cdot{}!a\cdot{}!b_1\cdots !b_n\cdot{} ?a
\end{array}
$$
lead to two configurations $\configuration_1',\configuration_2'$
with the same control state $q$ for peer $j$ as in the configuration
reached after $\tracebis_1$ or $\tracebis_2$.
On the other hand, for all $k\neq j$,
$\machineprojection{k}{\tracebis_1'}=\machineprojection{k}{\tracebis_2'}$, therefore
$\tracebis_1'\systemequiv\tracebis_2'$.
Since $\trace_0'\cdot{}!a\cdot{}!a_2\cdots!a_n\in\traces{n}{\system}$,
and $\machineprojection{i}{\trace_0'\cdot{}!a}=\machineprojection{i}{\tracebis_1'}=\machineprojection{i}{\tracebis_2'}$,
the two traces
$$
\begin{array}{ll}
\tracebis_1''= & \trace_0'\cdot{}!a\cdot{}?a\cdot{}!b_1\cdots!b_n\cdot{}!a_2\cdots!a_n\\
\tracebis_2''= & \trace_0'\cdot{}!a\cdot{}!b_1\cdots!b_n\cdot{}?a\cdot{}!a_2\cdots!a_n
\end{array}
$$
belong to $\traces{}{\system}$ and $\tracebis_1''\systemequiv\tracebis_2''$.
Consider first $\tracebis_1''$: this is $\normaltrace(\trace)$ as defined above,
therefore $\normaltrace(\trace)\in\traces{}{\system}$, and $\normaltrace(\trace)\systemequiv\tracebis_2''$. Consider now $\tracebis_2''$. By definition,
$\tracebis_2''\causalequiv\normaltrace(\trace')\cdot{}?a$. By hypothesis,
$\normaltrace(\trace')\systemequiv\trace'$, therefore $\normaltrace(\trace')\cdot{}?a\causalequiv \trace$. To sum up, $\normaltrace(\trace)\systemequiv\tracebis_2''\causalequiv\normaltrace(\trace')\cdot{}?a\systemequiv\trace$, therefore
$\normaltrace(\trace)\systemequiv\trace$.
\end{proof}

Note how we used the hypothesis that the communication topology is an oriented
ring in the proof of Lemma~\ref{lem:trace-normalization}.
As a hint that this trace normalization does not hold if a peer
can send to two different peers, consider the following example:
\begin{exa}
Let $\peer_1=!a\cdot{}!b$,
$\peer_2=?a$,
$\peer_3=?b$,
in other words, $\peer_1$ sends a message
to both $\peer_2$ and $\peer_3$, and they can receive this message.
The system is obviously $1$-synchronizable.
But the only trace $\systemequiv$-equivalent
to $\trace=!a\cdot{}!b\cdot{}?b$ is $\trace$ itself, which is not a
normalized trace.
\end{exa}

\ificalp\else
\subsection{Reachability set}
\fi

As a consequence of Lemma~\ref{lem:trace-normalization},
$1$-synchronizability implies several interesting properties
on the reachability set.

\begin{defi}[Channel-recognizable reachability set~\cite{Pachl87,CeceF05}]
Let
$\system=\systemlist{\peer_1,\dots,\peer_{\nbpeers}}$ with
$\peer_i=\tuple{\states_i,\transitions_i,q_{0,i}}$. The (coding of the) \emph{reachability set} of
$\system$ is the language
$\mathsf{Reach}(\system)$ over the alphabet $(M\cup\bigcup_{i=1}^{\nbpeers}\states_i)^*$
defined as $\{\state_1\cdots{}\state_{\nbpeers}\cdot{}w_1\cdots{}w_{\nbpeers}\mid
\configuration_0\tr{\trace}(\state_1,\dots,\state_{\nbpeers},w_1,\dots,w_{\nbpeers}),
\trace\in\traces{}{\system}\}$.
$\mathsf{Reach}(\system)$ is \emph{channel recognizable} (or QDD representable~\cite{BoigelotG99})
if it is a recognizable (and rational) language.
\end{defi}

\begin{thm}\label{thm:channel-rec}
Let $M$ be a message set such that $G_M$ is an oriented ring, and let
$\system$ be a $M$-system that is $1$-synchronizable.
Then
\begin{enumerate}
\item the reachability set of $\system$ is channel recognizable,
\item for all
$\trace\in\traces{}{\system}$, for all
$\configuration_0\tr{\trace}\configuration$, there is a stable
configuration $\configuration'$, $n\geq 0$ and $m_1,\dots m_n\in\messages$
such that $\configuration\tr{?m_1\cdots{}?m_n}\configuration'$.
\end{enumerate}
In particular, $\system$ does not have unspecified
receptions.
\end{thm}

\begin{proof}
\hfill
\begin{enumerate}
\item
Let $S$ be the set of stable configurations $\configuration$
such that $\configuration_0\tr{\trace}\configuration$ for some $\trace\in\traces{0}{\system}$; $S$ is finite and effective.
By Lemma~\ref{lem:trace-normalization},
$\mathsf{Reach}(\system)=\bigcup\{\mathsf{Reach}^!(\configuration)\mid\configuration\in S\}$, where
$\mathsf{Reach}^!(\configuration)=\{\state_1\cdots{}\state_{\nbpeers}\cdot{}w_1\cdots{}w_{\nbpeers}\mid
\configuration\tr{!a_1\cdots !a_n}(\state_1,\dots,\state_{\nbpeers},w_1,\dots,w_{\nbpeers}),n\geq 0,
a_1,\dots a_n\in\messages\}$ is an effective rational language.
\item Assume $\configuration_0\tr{\trace}\configuration$. By Lemma~\ref{lem:trace-normalization}, $\configuration_0\tr{\trace_0\cdot{}!m_1\cdots{}!m_r}\configuration$ for some
$\trace_0\in\traces{0}{\system}$. Then
$\trace_0\cdot{}!m_1\cdots{}!m_r\causalequiv\trace_0\cdot{}\trace_1$ where
$\trace_1:=!a_1\cdots{}!a_n\cdot{}b_1\cdots{}b_m$ for some
$a_1,\dots,a_n,b_1,b_m$ such that $\msgsrc(a_i)\neq\msgsrc(b_j)$ for
all $i\in\{1,\dots,n\}$ and $j\in\{1,\dots,m\}$. By
Lemma~\ref{lem:send-send-diagram-star},
$\trace_0\cdot{}\trace_1\cdot{}\overline{\trace_1}\in\traces{}{\system}$
(where $\overline{\trace_1}=?a_1\cdots{}?a_n\cdot{}?b_1\cdots{}?b_m$),
and therefore
$\configuration_0\tr{\trace_0\cdot{}\trace_1}\configuration\tr{\overline{\trace_1}}\configuration'$ for some stable configuration $\configuration'$.
\qedhere
\end{enumerate}
\end{proof}

\ificalp
\else
\subsection{$1$-synchronizability implies synchronizability}
\fi

\ificalp
\noindent\begin{minipage}{.8\textwidth} \setlength{\parfillskip}{0pt}
\begin{thm}\label{thm:weak-implies-strong-bipartite}
Let $M$ be a message set such that $G_M$ is an oriented ring. For any $M$-system $\system$,
$\system$ is
 $1$- synchronizable if and only if it is synchronizable.
\end{thm}

 In order to prove Theorem~\ref{thm:weak-implies-strong-bipartite}, we prove by induction on the length of $\trace$
\end{minipage}
\raisebox{1mm}{\begin{minipage}{0.1\textwidth}
\begin{tikzpicture}[scale=1.2]

\node (0) at (0,0) {};
\node (1) at (0,-1) {};
\node (2) at (1,-1) {};

\draw[->] (0) edge node [left] {$\trace$} (1);
\draw[->,dashed] (1) edge node [below] {$?m_1\cdots{}?m_k$} (2);
\draw[->,dashed] (0) edge node [above right] {$\synchtraceof{\trace}$} (2);

\end{tikzpicture}

\end{minipage}}
\\[-.1em]
that $\trace\cdot{}?m_1\cdots{}?m_k\systemequiv\synchtraceof{\trace}$
for some messages $m_1,\dots,m_k$, where
$\synchtraceof{\trace}$ denotes the unique synchronous $M$-trace such that
$\sendprojection{\synchtraceof{\trace}}=\sendprojection{\trace}$.
\else
\begin{thm}\label{thm:weak-implies-strong-bipartite}
Let $M$ be a message set such that $G_M$ is an oriented ring. For any $M$-system $\system$,
$\system$ is
 $1$- synchronizable if and only if it is synchronizable.
\end{thm}
\begin{proof}
We only need to show that $1$-synchronizability implies synchronizability. Let us
assume that $\system$ is $1$-synchronizable.
Let $\synchtraceof{\trace}$ denote the unique synchronous $M$-trace such that
$\sendprojection{\synchtraceof{\trace}}=\sendprojection{\trace}$.
We prove by induction on $\trace$ the following property (which implies in particular
that $\system$ is synchronizable): \\[1em]
\noindent\begin{minipage}{.75\textwidth} \setlength{\parfillskip}{0pt}
for all $\trace\in\traces{}{\system}$, there are $m_1,\dots,m_k\in\messages$
such that
$$
\begin{array}{cl}
(C1) & \synchtraceof{\trace}\in\traces{0}{\system}\\
(C2) & \trace\cdot{} ?m_1\cdots ?m_k\in\traces{}{\system},\mbox{ and}\\
(C3) & \trace\cdot{} ?m_1\cdots ?m_k\systemequiv\synchtraceof{\trace}.
\end{array}
$$
Let $\trace=\trace'\cdot{}\action$ be fixed and assume that
there are $m_1',\dots,m_k'\in\messages$ such that
$\trace'\cdot{}?m_1'\cdots?m_k'\in\traces{}{\system}$,
$\synchtraceof{\trace'}\in\traces{0}{\system}$, and
$\trace'\cdot{}?m_1'\cdots{}?m_k'\systemequiv\synchtraceof{\trace'}$.
Let us show that $(C1)$, $(C2)$, and $(C3)$ hold for $\trace$.
We reason by case analysis on the last action $\action$ of $\trace$.
\end{minipage}
\hspace{.05\textwidth}
\raisebox{0mm}{\begin{minipage}{0.2\textwidth}
\begin{tikzpicture}[scale=2]

\node (0) at (0,0) {};
\node (1) at (0,-1) {};
\node (2) at (1,-1) {};

\draw[->] (0) edge node [left] {$\trace$} (1);
\draw[->,dashed] (1) edge node [below] {$?m_1\cdots{}?m_k$} (2);
\draw[->,dashed] (0) edge node [above right] {$\synchtraceof{\trace}$} (2);

\end{tikzpicture}

\end{minipage}}
\begin{itemize}
\item Assume $\action=?a$.
Then $\synchtraceof{\trace}=\synchtraceof{\trace'}\in\traces{0}{\system}$, which proves
$(C1)$.
Let $i=\msgdst(a)$. Since peer $i$ only receives
on one channel, there are $m_1,\dots,m_{k-1}$
such that
$$
\trace'\cdot{}?m_1'\cdot{}?m_k'\causalequiv\trace'\cdot{}?a\cdot{}?m_1\cdot{}?m_{k-1}.
$$
Since $\trace'\cdot{}?m_1'\cdot{}?m_k'\systemequiv\synchtraceof{\trace}$ by
induction hypothesis, $(C2)$ and $(C3)$ hold.
\item Assume $\action=!a$. By Lemma~\ref{lem:trace-normalization},
there is $\normaltrace(\trace')=\trace_0\cdot{}!m_1''\cdots !m_k''$ with
$\trace_0\in\traces{0}{\system}$ such that $\trace'\systemequiv\normaltrace(\trace')$.
Since $\trace'\cdot{}?m_1'\cdots{}?m_k'$ leads to a stable configuration,
$m_1'',\dots,m_k''$ is a permutation of $m_1',\dots,m_k'$ that
do not swap messages of a same channel.
Since $G_M$ is an oriented ring,
$
\normaltrace(\trace')\causalequiv \trace_0\cdot{}!m_1'\cdots{}!m_k'.
$
Since $\trace'\cdot{} !a\in\traces{}{\system}$,
it holds that
$\trace_0\cdot{} !m_1\cdots !m_k\cdot{} !a\in\traces{}{\system}$, which implies by Lemma~\ref{lem:send-send-diagram-star} that the two traces
$$
\begin{array}{ll}
\tracebis_1= & \trace_0\cdot{} !m_1'\cdots !m_k'\cdots ?m_1\cdots?m_k'\cdot{}!a\cdot{}?a\\
\tracebis_2= & \trace_0\cdot{} !m_1'\cdots !m_k'\cdot{}!a\cdots ?m_1'\cdots?m_k'\cdot{}?a\\
\end{array}
$$
belong to $\traces{}{\system}$ and verify $\tracebis_1\systemequiv\tracebis_2$.
Consider first $\tracebis_1$, and let
$\tracebis_1'=\trace_0\cdot{} !m_1'\cdots !m_k'\cdot{}$ $?m_1'\cdots?m_k'$.
Since $\normaltrace(\trace')\causalequiv \trace_0\cdot{}!m_1'\cdots !m_k'\systemequiv\trace'$ and
$\trace'\cdot{}?m_1\cdots?m_k\systemequiv\synchtraceof{\trace'}$, it holds that
$\tracebis_1'\systemequiv\synchtraceof{\trace'}$.
Therefore,
$\synchtraceof{\trace'}\cdot{}!a\cdot{}?a=\synchtraceof{\trace}$ belongs to $\traces{}{\system}$, which shows $(C1)$, and $\synchtraceof{\trace}\systemequiv\tracebis_1$.
Consider now
$\tracebis_2$, and let $\tracebis_2'=\trace_0\cdot{} !m_1'\cdots !m_k'\cdot{}!a\causalequiv\normaltrace(\trace')\cdot{}!a$. Then $\tracebis_2'\systemequiv \trace'\cdot{} !a=\trace$,
therefore $\trace\cdot{}?m_1'\cdots?m_k'\cdot{}?a\in\traces{}{\system}$,
which shows $(C2)$,
and
$\trace\cdot{}?m_1'\cdots?m_k'\cdot{}?a\systemequiv\tracebis_2$. Since
$\tracebis_2\systemequiv\tracebis_1\systemequiv\synchtraceof{\trace}$, this
shows $(C3)$.
\qedhere
\end{itemize}
\end{proof}
\fi

\begin{thm}\label{thm:decidability}
Let $M$ be a message set such that $G_M$ is an oriented ring. The problem of deciding
whether a given $M$-system is synchronizable is
decidable.
\end{thm}
Since a system with two machines is a particular case of a system that is an oriented ring, we deduce from the above result that synchronizability is decidable in that particular case.
\begin{thm}
Synchronizability is decidable for systems of $2$-CFSMs.
\end{thm}

\section{Extensions and Related Works\label{sec:extensions}}
\subsection{Synchronizability for other communication models}

We considered the notions of synchronizability and language synchronizability
introduced by Basu and Bultan~\cite{BasuB16} and we
showed that both are not decidable for systems with peer-to-peer FIFO communications, called ($1$-$1$) type
systems in~\cite{BasuB16}.
In the same work,
Basu and Bultan considered the question of the decidability
of language synchronizability
for other communication
models. All the results we presented so far do not have any immediate consequences on their claims for these communication
models. Therefore, we briefly discuss now what we can say about the decidability of language synchronizability for the other communication
models that have been considered.

\subsubsection{Bags}
In~\cite{BasuB16}, language synchronizability is studied for systems where peers communicate through bags instead of queues,
thus allowing to reorder messages. Language synchronizability is decidable for bag communications:
$\tracesbag{}{\system}$ is the trace language of a Petri net, $T_0(\system)=\{\trace\in\actions{\messages}^*\mid\sendprojection{\trace}\in\sendtracesbag{0}{\system}\}$ is an effective
regular language, $\system$ is language synchronizable iff
$\tracesbag{}{\system}\subseteq T_0(\system)$,
and whether the trace language of a Petri is included in a given regular language
reduces to the coverability problem. Lossy communications where not considered in~\cite{BasuB16},
but the same kind of argument would also hold for lossy communications. However, our Example~\ref{ex:counter-example-peer-to-peer}
is a counter-example for Lemma~3 in~\cite{BasuB16}, \emph{i.e.} the notion of language $1$-synchronizability
for bag communications defined in~\cite{BasuB16} does not imply language synchronizability. The
question whether (language) synchronizability can be decided more efficiently than by reduction to the coverability problem
for Petri nets is open.

\subsubsection{Mailboxes}
The other communication models considered in~\cite{BasuB16} keep the FIFO queue model, but differ in the way how
queues are distributed among peers.
The $*$-$1$ (mailbox) model assumes a queue per receiver. This model
is the first model that was
considered for (language) synchronizability~\cite{BasuB11,BasuBO12}.
Our
Example~\ref{ex:counter-example-peer-to-peer} is not easy to adapt
for this communication model. We therefore
design a completely different counter-example.
\begin{exa}
\label{ex:counter-example-mailbox}
Consider the system of communicating machines depicted in
Fig.~\ref{fig:mailbox-counter-example}.
Assume that the machines communicate via mailboxes, like in~\cite{BasuB11,BasuBO12}, i.e. all messages that are sent to peer $i$
wait in a same FIFO queue $Q_i$, and let $\sendtracesmailbox{k}{\system}$ denote the $k$-bounded send
traces of $\system$ within this model of communications. Then
$\sendtracesmailbox{0}{\system}=\sendtracesmailbox{1}{\system}\neq\sendtracesmailbox{2}{\system}$, as depicted in Figure~\ref{fig:mailbox-counter-example}.
Therefore $\system$
is
language 1-synchronizable but not language
synchronizable, which contradicts Theorem~1 in~\cite{BasuB11},
Theorem~2 in~\cite{poplBasuBO12}, and Theorem~2 in~\cite{BasuB16}.
It can be noticed that it does not contradict Theorem~1 in~\cite{BasuBO12},
but it contradicts the Lemma~1 of the same paper, which is used to prove Theorem~1.
\end{exa}
\begin{figure}
\begin{tikzpicture}[shorten >=1pt,=stealth’,initial text={},auto,every state/.style={scale=.4,initial distance={2mm}},scale=1.2]
{\small
\begin{scope}
\node[state,initial] (m1q0) at (0,0) {};
\node[state] (m1q1) at (1,0) {};
\node[state] (m1q2) at (2,0) {};
\node[state] (m1q3) at (3,0) {};
\node[left of=m1q0] {$\peer_1$};
\draw[->] (m1q0) edge node {$!\tagmsg{a}{1}{2}$} (m1q1);
\draw[->] (m1q1) edge node {$!\tagmsg{a}{1}{2}$} (m1q2);
\draw[->] (m1q2) edge node {$!\tagmsg{b}{1}{3}$} (m1q3);
\end{scope}

\begin{scope}[yshift=-.8cm]
\node[state,initial] (m2q0) at (0,-1) {};
\node[state] (m2q1) at (1,-1) {};
\node[state] (m2q2) at (2,-.5) {};
\node[state] (m2q3) at (2,-1.5) {};
\node[state] (m2q4) at (3,-1) {};
\node[state] (m2q5) at (4,-1) {};
\node[state] (m2q6) at (0,0) {};
\node[state] (m2q7) at (1,0) {};
\node[state] (m2q8) at (2,0) {};
\node[state] (m2q9) at (3,0) {};
\node[state] (m2q10) at (4,0) {};
\node[left of=m2q0] {$\peer_2$};
\draw[->] (m2q0) edge node {$?\tagmsg{a}{1}{2}$} (m2q1);
\draw[->] (m2q1) edge node [above] {$?\tagmsg{a}{1}{2}$} (m2q2);
\draw[->] (m2q1) edge node [below left] {$!\tagmsg{c}{2}{3}$} (m2q3);
\draw[->] (m2q2) edge node {$!\tagmsg{c}{2}{3}$} (m2q4);
\draw[->] (m2q3) edge node [below right] {$?\tagmsg{a}{1}{2}$} (m2q4);
\draw[->] (m2q4) edge node {$?\tagmsg{d}{3}{2}$} (m2q5);
\draw[->] (m2q0) edge node {$!\tagmsg{c}{2}{3}$} (m2q6);
\draw[->] (m2q6) edge node {$?\tagmsg{a}{1}{2}$} (m2q7);
\draw[->] (m2q7) edge node {$?\tagmsg{a}{1}{2}$} (m2q8);
\draw[->] (m2q8) edge node {$?\tagmsg{d}{3}{2}$} (m2q9);
\draw[->] (m2q9) edge node {$!\tagmsg{e}{2}{1}$} (m2q10);
\end{scope}

\begin{scope}[xshift=5cm]
\node[state,initial] (m3q0) at (0,0) {};
\node[state] (m3q1) at (0,-.8) {};
\node[state] (m3q2) at (0,-1.6) {};
\node[state] (m3q3) at (1,0) {};
\node[state] (m3q4) at (1,-.8) {};
\node[state] (m3q5) at (1,-1.6) {};
\node[left of=m3q0] {$\peer_3$};
\draw[->] (m3q0) edge node [left] {$?\tagmsg{c}{2}{3}$} (m3q1);
\draw[->] (m3q1) edge node [left] {$?\tagmsg{b}{1}{3}$} (m3q2);
\draw[->] (m3q0) edge node {$?\tagmsg{b}{1}{3}$} (m3q3);
\draw[->] (m3q3) edge node [left] {$?\tagmsg{c}{2}{3}$} (m3q4);
\draw[->] (m3q4) edge node [left] {$!\tagmsg{d}{3}{2}$} (m3q5);
\end{scope}
}



\node at (8.5cm,-.8cm) {\small
$
\begin{array}{llll}
\sendtracesmailbox{0}{\system} & = & \prefixclosure \{ &
aabcd,\\ &&& aacb,\\ &&& acab, \\ &&& caab \}
\\
& = & \multicolumn{2}{l}{\sendtracesmailbox{1}{\system}}
\\
\end{array}
$
};

\node at (8cm,-2cm) {\small
$\sendtracesmailbox{2}{\system} =  \sendtracesmailbox{0}{\system}\cup \{aabcde\}$
};

\end{tikzpicture}
\caption{\label{fig:mailbox-counter-example}Language $1$-synchronizability does not imply
language synchronizability for $1$-$*$ (mailbox) communications \emph{à la}~\cite{BasuB11,BasuBO12}.}
\end{figure}



\subsection{Analysis of the original mistake\label{sec:error-explaned}}

We analyse the original mistake looking at the proof of Theorem~1 in~\cite{BasuB11}.
The proof attempt is by absurd: the authors assume a sequence of send actions $m_1\dots m_n$
that exists in $\sendtraces{}{S}$ but not in $\sendtraces{1}{S}$.
There exists a prefix $m_1\dots m_l$ in $\sendtraces{1}{S}$ such that
$m_1\dots m_{l+1}\not\in\sendtraces{1}{S}$. So there are two traces
$\trace\in\traces{}{S}$ and $\trace'\in\traces{1}{S}$ with
$\sendprojection{\trace}=m_1\dots m_{l+1}$ and
$\sendprojection{\trace'}=m_1\dots m_{l}$. The authors claim that the only reason why
$\trace'$ cannot be extended (in $\traces{1}{S}$) to a trace that ends with $!m_{l+1}$
is because the buffer where $m_{l+1}$ should go is full. But they miss another explanation: it could
simply be that the configuration after $\trace'$ has control states from which it is not possible to take
a transition labeled with a $!m_{l+1}$, even after a few receptions. This configuration has a priori nothing in common with the configuration reached in $\trace$ right before $!m_{l+1}$.

\subsection{Realizability of choreographies}
Let us recall that a choreography $C$ is a finite automaton describing the exchange of messages between processes. A transition $(q, \tagmsg{m}{i}{j}, q')$ in $C$ is interpreted as follows: process $P_i$, in state $q$, sends message $m$ to process $P_j$ and moves to state $q'$; and in the same way, process $P_j$, in state $q$, receives message $m$ from process $P_i$ and moves into state $q'$. The communication has to be specified and can be done by rendez-vous, bags, fifo channels ; the topology of communications could be peer-to-peer or with mailboxes.
From a choreography $C$, one may construct the system $S_C$ of communicating processes $P_i$ such that
each process $P_i$ is the (natural) projection of $C$ ; then $C$ coincides with the synchronous composition of the peer-to-peer system of $P_i$ (Proposition 4 in \cite{schewe2020realisability}). But choreography-defined peer-to-peer systems form a \emph{strict subclass} of peer-to-peer systems.

Since the word \emph{realizability} is used with different meanings, for example in ~\cite{poplBasuBO12} and in ~\cite{schewe2020realisability}, we distinguish here two notions of realizability. A choreography $C$ is said \emph{mailbox-realizable} (resp. \emph{peer-to-peer-realizable}) if the system $S_C$ with respect to the mailbox semantics (resp. with respect to the peer-to-peer semantics) is synchronizable.

Basu, Bultan and Ouederni considered the question of the decidability of the mailbox-realizability of
choreographies~\cite{poplBasuBO12}.
Assuming (from a previous paper from Basu and Bultan \cite{BasuB11}) that
$\sendtraces{0}{S}=\sendtraces{1}{S}$ implies $\sendtraces{0}{S}=\sendtraces{}{S}$, they established
the decidability of the mailbox-realizability of choreographies. Our counter-example shows that this
decidability proof is not correct hence the decidability of the mailbox-realizability is, to the best of our knowledge, still an open problem. They did not studied the peer-to-peer-realizability problem.

Very recenly, Schewe~\emph{et al}~\cite{schewe2020realisability}
considered the peer-to-per-realizability problem and proposed a proof of decidability noticing that all our counter-examples are not choreography-defined peer-to-peer systems. They did not studied the mailbox-realizability problem.

\subsection{Branching synchronizability and stability} Branching synchronizability is defined in \cite{DBLP:conf/facs2/OuederniSB13} and Theorem 1 says that a system $S$ of processes communicating through fifo channels and mailboxes is branching synchronizable iff its associated synchronous system $S_{rdv}$ is branching equivalent (i.e. bisimilar) to $S$ in which all channels are bounded by $1$. It is immediate to deduce from Theorem 1 that branching synchronizability is decidable but this is false. The proof of Theorem 1 is not given in \cite{DBLP:conf/facs2/OuederniSB13} and it is said that it is on the web page of the first author, Ouederni; we did not found the complete paper on her web pages. 
Stability~\cite{AkrounSY16} seems to be another name for branching synchronizability.
More precisely, let $\ltssendof{k}{\system}$ denote the labeled transition system
restricted to $k$-bounded configurations, where receive actions
are considered as internal actions ($\tau$ transitions in CCS dialect).
A system $\system$ is $k$-stable if
$\ltssendof{}{\system}\branchsim\ltssendof{k}{\system}$,
where $\branchsim$ denotes the branching bisimulation.
In particular, a system that is $0$-stable is synchronizable.
Theorem~1 in~\cite{AkrounSY16} claims that the following
implication would hold for any $k\geq 1$: if
$\ltssendof{k}{\system}\branchsim\ltssendof{k+1}{\system}$, then $\ltssendof{k+1}{\system}\branchsim\ltssendof{k+2}{\system}$. Our
example~\ref{ex:counter-example-peer-to-peer} is a counter-example to this
implication for $k=0$, and it could be generalized to a counter-example for
other values of $k$ by changing the number of consecutive $a$ messages
that are sent by the first peer (and, symmetrically, received by the second peer). Therefore the claim of Theorem~1 in~\cite{AkrounSY16} is not correct.

In~\cite{AkrounS18}, the authors consider the $\ltsof{k}{\system}$ (note the ``$?$'') associated
with a given system: this LTS is the ``standard'' one that
keeps the receive actions as being ``observable''.
A new notion,
also called stability is defined accordingly: a system (strongly) $k$-stable
if $\ltsof{}{\system}\branchsim\ltsof{k}{\system}$, and (strongly) stable
if it is strongly $k$-stable for some $k$. It is not difficult to observe
that a system is strongly $k$-stable if and only if
all its traces are $k$-bounded:
indeed, if all traces are
$k$-bounded, $\ltsof{}{\system}=\ltsof{k}{\system}$,
and if not, there is a trace with $k+1$ unmatched send actions in
$\ltsof{}{\system}$, therefore $\ltsof{}{\system}$ is not trace
equivalent to $\ltsof{k}{\system}$. All results of~\cite{AkrounS18}
are therefore trivially correct.

\subsection{\label{sec:existentially-bounded}Existentially bounded systems}

Existentially bounded systems have been introduced by Genest, Kuske and
Muscholl~\cite{GenestKM06}.
A system $\system$ is existentially $k$-bounded, $k\geq 1$,
if for all trace $\trace\in\traces{}{\system}$, there
is a trace $\trace'\in\traces{k}{\system}$ such that
$\trace\causalequiv\trace'$. Unlike synchronizability,
existential boundedness takes into account the receive actions, but bases on
a more relaxed notion of trace
(also called message sequence chart, MSC for short).

Existential boundedness and synchronizability are incomparable. For instance,
a system with two peers $P_1$ and $P_2$, defined (in CCS notation) as
$P_1=!a$ and $P_2=0$ (idle), is existentially $1$-bounded, but not
synchronizable. Conversely, there are synchronous systems that are not
existentially $1$-bounded: consider
$P=!a.!a||?b.?b$ (i.e. all shuffles of the two),
and $Q=?a.?a||!b!b$, and assume that $P,Q$ are represented as
(single-threaded) communicating automata.
Then this system is
synchronous, but the trace $!a!a!b!b?a?a?b?b$ is not causally equivalent to a
$1$-bounded trace.

Although Genest~\emph{et al} did not explicitly defined it, one could
consider existentially $0$-bounded systems. This is a quite restricted notion,
but it would imply synchronizability and would generalize
half-duplex systems.

Genest~\emph{et al} showed that for any given $k\geq 1$, it is decidable
whether a given system $\system$ of communicating machines with peer-to-peer
communications is existentially $k$-bounded
(Proposition 5.5, \cite{GenestKM10}). Note that what we call a system is
what Genest~\emph{et al} called a \underline{deadlock-free} system, since
we do not have any notion of accepting states.

\subsection{Communication layers}

Finally, following the work of Lipton on reduction~\cite{Lipton75},
there has been recently a lot of interest on the verification of FIFO systems
on the idea of grouping communications in closed
rounds~\cite{Chaouch-SaadCM09,KraglQH18},
in particular to abstract a round of communications
as a single operation.
In~\cite{Bouajjani-el-al-CAV-2018},
the authors define the notion of $k$-synchronous systems:
a system $\system$ of machines communicating with mailboxes is $k$-synchronous
if for all $\trace\in\traces{k}{\system}$, there are
$\trace_1,\ldots,\trace_n$ such that
\begin{itemize}
  \item $\trace\causalequiv\trace_1\cdots\trace_n$,
  \item for all $i=1,\ldots,n$ $\trace_i$ contains at mots $k$ send actions, and
  \item every message received in $\trace_i$ has been sent in $\trace_i$
\end{itemize}
The classes of $k$-synchronous systems, of existentially $k$ bounded systems, and the one of
synchronizable systems are incomparable, although they share very similar ideas.

\section{\label{sec:open}Conclusion and Perspectives}
We established the undecidability of synchronizability for communicating finite state machines communicating with peer-to-peer channels. We also proposed a counter-example for an argument of the proofs
that synchronizability is decidable for mailbox communications. Finally, we showed the decidability of
synchronizability for systems organized on an oriented ring.

Although we identified some problems and fixed them, our work leaves open a bunch of questions. The first one is the decidability of synchronizability for the mailboxes semantics - we only found a counter example to the proof of Basu and Bultan, but we did not show that it is undecidable.
Another question is the decidability of the LTL/CTL model checking for synchronizable systems, either on traces, or on sequences of configurations. We also left open the exact complexity of synchronizability for oriented rings. We believe these questions are rather technical and sometimes very challenging.

\subsection*{\label{sec:ack}Acknowledgements}
We thank the anonymous reviewers of the Conference ICALP 2017 and of the Journal LMCS who produced detailed and usefull reviews that helped and motivated us to produce a better version of this paper. In particular, we added many explanations of the technical proofs, we gave more understandable details and we created Section 5 devoted to explain mistakes in previous decidability proofs and results. We also thank Shrisha Rao for his carrefull reading of the paper.

\bibliographystyle{alphaurl}
\bibliography{refs}

\begin{thebibliography}{VVGK10}

\bibitem[AS18]{AkrounS18}
Lakhdar Akroun and Gwen Sala{\"{u}}n.
\newblock Automated verification of automata communicating via {FIFO} and bag
  buffers.
\newblock {\em Formal Methods in System Design}, 52(3):260--276, 2018.

\bibitem[ASY16]{AkrounSY16}
Lakhdar Akroun, Gwen Sala{\"{u}}n, and Lina Ye.
\newblock Automated analysis of asynchronously communicating systems.
\newblock In {\em SPIN'16}, pages 1--18, 2016.

\bibitem[BB11]{BasuB11}
Samik Basu and Tevfik Bultan.
\newblock Choreography conformance via synchronizability.
\newblock In {\em Procs. of {WWW} 2011}, pages 795--804, 2011.
\newblock \href {https://doi.org/10.1145/1963405.1963516}
  {\path{doi:10.1145/1963405.1963516}}.

\bibitem[BB16]{BasuB16}
Samik Basu and Tevfik Bultan.
\newblock On deciding synchronizability for asynchronously communicating
  systems.
\newblock {\em Theor. Comput. Sci.}, 656:60--75, 2016.
\newblock \href {https://doi.org/10.1016/j.tcs.2016.09.023}
  {\path{doi:10.1016/j.tcs.2016.09.023}}.

\bibitem[BBO12a]{poplBasuBO12}
Samik Basu, Tevfik Bultan, and Meriem Ouederni.
\newblock Deciding choreography realizability.
\newblock In {\em Procs. of POPL'12}, pages 191--202, 2012.
\newblock \href {https://doi.org/10.1145/2103656.2103680}
  {\path{doi:10.1145/2103656.2103680}}.

\bibitem[BBO12b]{BasuBO12}
Samik Basu, Tevfik Bultan, and Meriem Ouederni.
\newblock Synchronizability for verification of asynchronously communicating
  systems.
\newblock In {\em Procs. of {VMCAI} 2012}, 2012.
\newblock \href {https://doi.org/10.1007/978-3-642-27940-9_5}
  {\path{doi:10.1007/978-3-642-27940-9_5}}.

\bibitem[BEJQ18]{Bouajjani-el-al-CAV-2018}
Ahmed Bouajjani, Constantin Enea, Kailiang Ji, and Shaz Qadeer.
\newblock On the completeness of verifying message passing programs under
  bounded asynchrony.
\newblock In {\em CAV 2018}, pages 372--391, 2018.

\bibitem[BG99]{BoigelotG99}
Bernard Boigelot and Patrice Godefroid.
\newblock Symbolic verification of communication protocols with infinite state
  spaces using qdds.
\newblock {\em Formal Methods in System Design}, 14(3):237--255, 1999.
\newblock \href {https://doi.org/10.1023/A:1008719024240}
  {\path{doi:10.1023/A:1008719024240}}.

\bibitem[BZ81]{BrandZ81}
Daniel Brand and Pitro Zafiropulo.
\newblock On communicating finite-state machines.
\newblock Technical Report 1053, Tech. Rep. RZ, IBM Zurich Research Lab.,
  Ruschlikon, Switzerland, January 1981.

\bibitem[BZ83]{BrandZ}
Daniel Brand and Pitro Zafiropulo.
\newblock On communicating finite-state machines.
\newblock {\em J. ACM}, 30(2):323--342, April 1983.
\newblock \href {https://doi.org/10.1145/322374.322380}
  {\path{doi:10.1145/322374.322380}}.

\bibitem[CCM09]{Chaouch-SaadCM09}
Mouna Chaouch{-}Saad, Bernadette Charron{-}Bost, and Stephan Merz.
\newblock A reduction theorem for the verification of round-based distributed
  algorithms.
\newblock In Olivier Bournez and Igor Potapov, editors, {\em Reachability
  Problems, 3rd International Workshop, {RP} 2009, Palaiseau, France, September
  23-25, 2009. Proceedings}, volume 5797 of {\em Lecture Notes in Computer
  Science}, pages 93--106. Springer, 2009.
\newblock \href {https://doi.org/10.1007/978-3-642-04420-5\_10}
  {\path{doi:10.1007/978-3-642-04420-5\_10}}.

\bibitem[CF05]{CeceF05}
Gerald C{\'e}c{\'e} and Alain Finkel.
\newblock Verification of programs with half-duplex communication.
\newblock {\em Inf. Comput.}, 202(2):166--190, 2005.
\newblock \href {https://doi.org/10.1016/j.ic.2005.05.006}
  {\path{doi:10.1016/j.ic.2005.05.006}}.

\bibitem[CHS14]{ClementeHS14}
Lorenzo Clemente, Fr{\'{e}}d{\'{e}}ric Herbreteau, and Gr{\'{e}}goire Sutre.
\newblock Decidable topologies for communicating automata with {FIFO} and bag
  channels.
\newblock In {\em Procs. of {CONCUR} 2014}, pages 281--296, 2014.
\newblock \href {https://doi.org/10.1007/978-3-662-44584-6_20}
  {\path{doi:10.1007/978-3-662-44584-6_20}}.

\bibitem[CS08]{ChambartS08}
Pierre Chambart and Philippe Schnoebelen.
\newblock Mixing lossy and perfect fifo channels.
\newblock In {\em Procs. of {CONCUR} 2008}, pages 340--355, 2008.
\newblock \href {https://doi.org/10.1007/978-3-540-85361-9_28}
  {\path{doi:10.1007/978-3-540-85361-9_28}}.

\bibitem[DY12]{DenielouY12}
Pierre{-}Malo Deni{\'{e}}lou and Nobuko Yoshida.
\newblock Multiparty session types meet communicating automata.
\newblock In {\em Procs. of {ESOP} 2012}, pages 194--213, 2012.
\newblock \href {https://doi.org/10.1007/978-3-642-28869-2_10}
  {\path{doi:10.1007/978-3-642-28869-2_10}}.

\bibitem[FBS05]{DBLP:journals/tse/FuBS05}
Xiang Fu, Tevfik Bultan, and Jianwen Su.
\newblock Synchronizability of conversations among web services.
\newblock {\em {IEEE} Trans. Software Eng.}, 31(12):1042--1055, 2005.
\newblock \href {https://doi.org/10.1109/TSE.2005.141}
  {\path{doi:10.1109/TSE.2005.141}}.

\bibitem[GKM06]{GenestKM06}
Blaise Genest, Dietrich Kuske, and Anca Muscholl.
\newblock A kleene theorem and model checking algorithms for existentially
  bounded communicating automata.
\newblock {\em Inf. Comput.}, 204(6):920--956, 2006.
\newblock \href {https://doi.org/10.1016/j.ic.2006.01.005}
  {\path{doi:10.1016/j.ic.2006.01.005}}.

\bibitem[GKM10]{GenestKM10}
Blaise Genest, Dietrich Kuske, and Anca Muscholl.
\newblock On communicating automata with bounded channels.
\newblock {\em Fundamenta Informaticae}, 2010.

\bibitem[HGS12]{HeussnerGS12}
Alexander Heu{\ss}ner, Tristan~Le Gall, and Gr{\'{e}}goire Sutre.
\newblock Mcscm: {A} general framework for the verification of communicating
  machines.
\newblock In {\em Procs. of {TACAS} 2012}, pages 478--484, 2012.
\newblock \href {https://doi.org/10.1007/978-3-642-28756-5_34}
  {\path{doi:10.1007/978-3-642-28756-5_34}}.

\bibitem[HLMS12]{Heussneretal12}
Alexander Heu{\ss}ner, J{\'{e}}r{\^{o}}me Leroux, Anca Muscholl, and
  Gr{\'{e}}goire Sutre.
\newblock Reachability analysis of communicating pushdown systems.
\newblock {\em Logical Methods in Computer Science}, 8(3), 2012.
\newblock \href {https://doi.org/10.2168/LMCS-8(3:23)2012}
  {\path{doi:10.2168/LMCS-8(3:23)2012}}.

\bibitem[KQH18]{KraglQH18}
Bernhard Kragl, Shaz Qadeer, and Thomas~A. Henzinger.
\newblock Synchronizing the asynchronous.
\newblock In Sven Schewe and Lijun Zhang, editors, {\em 29th International
  Conference on Concurrency Theory, {CONCUR} 2018, September 4-7, 2018,
  Beijing, China}, volume 118 of {\em LIPIcs}, pages 21:1--21:17. Schloss
  Dagstuhl - Leibniz-Zentrum fuer Informatik, 2018.
\newblock \href {https://doi.org/10.4230/LIPIcs.CONCUR.2018.21}
  {\path{doi:10.4230/LIPIcs.CONCUR.2018.21}}.

\bibitem[Lip75]{Lipton75}
Richard~J. Lipton.
\newblock Reduction: {A} method of proving properties of parallel programs.
\newblock {\em Commun. {ACM}}, 18(12):717--721, 1975.
\newblock \href {https://doi.org/10.1145/361227.361234}
  {\path{doi:10.1145/361227.361234}}.

\bibitem[LMP08]{TorreMP08}
Salvatore {La Torre}, Parthasarathy Madhusudan, and Gennaro Parlato.
\newblock Context-bounded analysis of concurrent queue systems.
\newblock In {\em Procs. of {TACAS} 2008}, pages 299--314, 2008.
\newblock \href {https://doi.org/10.1007/978-3-540-78800-3_21}
  {\path{doi:10.1007/978-3-540-78800-3_21}}.

\bibitem[LP98]{Lewis-Papadimitriou}
Harry~R. Lewis and Christos~H. Papadimitriou.
\newblock {\em Elements of the theory of computation, 2nd Edition}.
\newblock Prentice Hall, 1998.

\bibitem[MM98]{Manohar1998}
Rajit Manohar and Alain~J. Martin.
\newblock Slack elasticity in concurrent computing.
\newblock {\em Mathematics of Program Construction}, pages 272--285, 1998.
\newblock \href {https://doi.org/10.1007/BFb0054295}
  {\path{doi:10.1007/BFb0054295}}.

\bibitem[OSB13]{DBLP:conf/facs2/OuederniSB13}
Meriem Ouederni, Gwen Sala{\"{u}}n, and Tevfik Bultan.
\newblock Compatibility checking for asynchronously communicating software.
\newblock In Jos{\'{e}}~Luiz Fiadeiro, Zhiming Liu, and Jinyun Xue, editors,
  {\em Formal Aspects of Component Software - 10th International Symposium,
  {FACS} 2013, Nanchang, China, October 27-29, 2013, Revised Selected Papers},
  volume 8348 of {\em Lecture Notes in Computer Science}, pages 310--328.
  Springer, 2013.
\newblock \href {https://doi.org/10.1007/978-3-319-07602-7\_19}
  {\path{doi:10.1007/978-3-319-07602-7\_19}}.

\bibitem[Pac87]{Pachl87}
Jan Pachl.
\newblock Protocol description and analysis based on a state transition model
  with channel expressions.
\newblock In {\em Proc. of Protocol Specification, Testing, and Verification,
  VII}, 1987.

\bibitem[SAAB20]{schewe2020realisability}
Klaus-Dieter Schewe, Yamine A{\"\i}t-Ameur, and Sarah Benyagoub.
\newblock Realisability of choreographies.
\newblock In {\em International Symposium on Foundations of Information and
  Knowledge Systems}, pages 263--280. Springer, 2020.

\bibitem[Sie05]{Siegel05}
Stephen~F. Siegel.
\newblock Efficient verification of halting properties for {MPI} programs with
  wildcard receives.
\newblock In {\em Procs. of {VMCAI} 2005}, pages 413--429, 2005.
\newblock \href {https://doi.org/10.1007/978-3-540-30579-8_27}
  {\path{doi:10.1007/978-3-540-30579-8_27}}.

\bibitem[VVGK10]{Vakkalanka2010}
Sarvani Vakkalanka, Anh Vo, Ganesh Gopalakrishnan, and Robert~M. Kirby.
\newblock Precise dynamic analysis for slack elasticity: Adding buffering
  without adding bugs.
\newblock In Rainer Keller, Edgar Gabriel, Michael Resch, and Jack Dongarra,
  editors, {\em Recent Advances in the Message Passing Interface}, pages
  152--159, Berlin, Heidelberg, 2010. Springer Berlin Heidelberg.
\newblock \href {https://doi.org/10.1007/978-3-642-15646-5_16}
  {\path{doi:10.1007/978-3-642-15646-5_16}}.

\end{thebibliography}

\end{document}